\documentclass[11pt]{article}

\usepackage{fullpage}
\usepackage{natbib}

\usepackage{latexsym,amssymb,amsmath}
\usepackage{epsfig}
\usepackage{math-cmds,math-envs,latin-abbrevs,xspace}
\usepackage{verbatim}
\usepackage{proof}
\usepackage{code}

\newcommand{\nostar}[1]{}
\newcommand{\noboxunbox}[1]{}

\renewcommand{\paragraph}[1]{{\bf {#1}}}

\newcommand{\codecolsep}{1ex}

\newcommand{\rlabel}[1]{\mbox{\small{\bf ({#1})}}}

\newcommand{\ttt}[1]{\texttt{#1}}

\newcommand{\out}[1] {}

\newcounter{codeLineCntr}

\newenvironment{codeListingNormal}
 {\setcounter{codeLineCntr}{0}
  \vspace{-.1in}
  \ttfamily\begin{tabbing}}
  {\end{tabbing}
 \vspace{-.1in}
}

\newenvironment{codeListing}
 {\setcounter{codeLineCntr}{0}
  \fontsize{10}{11}
  \vspace{-.1in}
  \ttfamily\begin{tabbing}}
  {\end{tabbing}
\vspace{-.1in}
}

\newenvironment{codeListingS}
 {\setcounter{codeLineCntr}{0}
  \fontsize{9}{9.6}
  \vspace{-.1in}
  \ttfamily\begin{tabbing}}
  {\end{tabbing}
\vspace{-.1in}
}

\reversemarginpar
\newif\ifnotes
\notestrue

\newcommand{\punt}[1]{}

\newcommand{\secref}[1]{Section~\ref{sec:#1}}

\newcommand{\figref}[1]{Figure~\ref{fig:#1}}

\newcommand{\thmref}[1]{Theorem~\ref{thm:#1}}

\renewcommand{\eqref}[1]{Equation~(\ref{eq:#1})}

\newcommand{\proc}[1]{\ifmmode\mbox{\textsc{#1}}\else\textsc{#1}\fi}

\newcommand{\func}[1]{\ifmmode\mathrm{#1}\else\textrm{#1}\fi}

\newcommand{\keywordaccent}[1]{{#1}}
\newcommand{\kfun}{\keywordaccent{fun}\xspace}
\newcommand{\kmfun}{\keywordaccent{mfun}\xspace}
\newcommand{\klet}{\keywordaccent{let}\xspace}

\newcommand{\kin}{\keywordaccent{in}\xspace}
\newcommand{\kend}{\keywordaccent{end}\xspace}

\newcommand{\kif}{\keywordaccent{if}\xspace}
\newcommand{\kmif}{\keywordaccent{mif}\xspace}
\newcommand{\kthen}{\keywordaccent{then}\xspace}
\newcommand{\kelse}{\keywordaccent{else}\xspace}
\newcommand{\kreturn}{\keywordaccent{return}\xspace}

\newcommand{\kcase}{\keywordaccent{case}\xspace}

\newcommand{\kof}{\keywordaccent{of}\xspace}

\newcommand{\kand}{\keywordaccent{and}\xspace}
\newcommand{\knil}{\keywordaccent{nil}\xspace}
\newcommand{\kNIL}{\keywordaccent{NIL}\xspace}
\newcommand{\kcons}{\keywordaccent{cons}}
\newcommand{\kCONS}{\keywordaccent{CONS}}

\newcommand{\update}[3]{\mathcd{update\cdparens{\ensuremath{{#1},{#2},{#3}}}}}

\newcommand{\ot}[2]{\ensuremath{~;~}}

\newcommand{\emptyContext}{\emptyset}

\newcommand{\ms}{\ensuremath{\sigma}}
\newcommand{\msp}{\ms'}
\newcommand{\mspp}{\ms''}
\newcommand{\msi}{\ms_1}
\newcommand{\msii}{\ms_2}

\newcommand{\MS}{\ensuremath{\Sigma}}
\newcommand{\MSP}{\MS'}
\newcommand{\MSPP}{\MS''}
\newcommand{\MSi}{\MS_1}
\newcommand{\MSii}{\MS_2}

\newcommand{\dom}[1]{\mathop{\textrm{dom}}(\ensuremath{#1})}

\newcommand{\br}{\beta}
\newcommand{\ev}{\varepsilon}

\newcommand{\bangev}[1]{\mathcd{!}{#1}}
\newcommand{\inlev}{\mathcd{inl}}
\newcommand{\inrev}{\mathcd{inr}}

\newcommand{\nullbr}{\bullet}
\newcommand{\consbr}[2]{{#1}\cdot{#2}}
\newcommand{\extbr}[2]{{#1}\mathop{\widehat{\ }}{#2}}

\newcommand{\bangbr}[2]{\extbr{#1}{\bangev{#2}}}
\newcommand{\inlbr}[1]{\extbr{#1}{\inlev}}
\newcommand{\inrbr}[1]{\extbr{#1}{\inrev}}

\newcommand{\atbr}[2]{{#1}\mathbin{@}{#2}}

\newcommand{\gbr}{\gamma}
\newcommand{\gev}{\epsilon}

\newcommand{\callgev}[1]{\cdparens{#1}}
\newcommand{\banggev}[1]{\mathcd{!}{#1}}
\newcommand{\inlgev}[1]{\mathcd{inl}\cdparens{#1}}
\newcommand{\inrgev}[1]{\mathcd{inr}\cdparens{#1}}
\newcommand{\pairgev}[2]{\langle #1,#2\rangle}

\newcommand{\nullgbr}{\nullbr}
\newcommand{\consgbr}[2]{\consbr{#1}{#2}}
\newcommand{\extgbr}[2]{\extbr{#1}{#2}}

\newcommand{\callgbr}[2]{\extgbr{#1}{\callgev{#2}}}
\newcommand{\banggbr}[2]{\extgbr{#1}{\banggev{#2}}}
\newcommand{\splitgbr}[3]{\extgbr{#1}{\pairgev{#2}{#3}}}
\newcommand{\inlgbr}[2]{\extgbr{#1}{\inlgev{#2}}}
\newcommand{\inrgbr}[2]{\extgbr{#1}{\inrgev{#2}}}

\newcommand{\atgbr}[2]{{#1}\mathbin{@}{#2}}

\newcommand{\simp}[1]{{#1}^{\circ}}

\newcommand{\reduces}[1]{\Downarrow^{#1}}
\newcommand{\treduces}{\mathrel{\reduces{\mathsf{t}}}}
\newcommand{\ereduces}{\mathrel{\reduces{\mathsf{e}}}}
\newcommand{\treducespure}{\mathrel{\reduces{\mathsf{t}}_{\mathsf{p}}}}
\newcommand{\ereducespure}{\mathrel{\reduces{\mathsf{e}}_{\mathsf{p}}}}

\newcommand{\tsi}{\ensuremath{\ts}}

\newcommand{\tso}{\ensuremath{\ts_{o}}}
 
\newcommand{\eis}[4]{{#1},{#2{:}#3},{#4}}
\newcommand{\eos}[2]{{#1},{#2}}

\newcommand{\tis}[2]{{#1},{#2}}
\newcommand{\tos}[2]{{#1},{#2}}

\newcommand{\MFL}{\textsf{MFL}\xspace}
\newcommand{\mfl}{\MFL}

\newcommand{\taui}{\tau_1}
\newcommand{\tauii}{\tau_2}

\renewcommand{\u}{u}

\renewcommand{\circle}{\bigcirc}

\newcommand{\ra}{\rightarrow}

\renewcommand{\t}{t}
\newcommand{\ti}{t_1}
\newcommand{\tii}{t_2}

\newcommand{\e}{e}
\newcommand{\ei}{e_1}
\newcommand{\eii}{e_2}

\renewcommand{\v}{v}
\newcommand{\vi}{v_1}
\newcommand{\vii}{v_2}
\newcommand{\vp}{v'}

\newcommand{\osum}{+}
\newcommand{\tunit}{1}
\newcommand{\uno}{\star}

\newcommand{\rt}[2]{\mu {#1}.{#2}}
\newcommand{\roll}[1]{\mathcd{roll}\cdparens{#1}}
\newcommand{\unroll}[1]{\mathcd{unroll}\cdparens{#1}}

\newcommand{\cdparens}[1]{\mathcd{(}{#1}\mathcd{)}}

\newcommand{\cdabracks}[1]{\ensuremath{\langle{#1}\rangle}}

\newcommand{\cdcomma}{\mathcd{,}}
\newcommand{\cdcolon}{\mathcd{:}}

\newcommand{\Integer}{\mathcd{int}}
\newcommand{\Int}{\Integer}

\newcommand{\Plus}{\mathcd{+}}
\newcommand{\Minus}{\mathcd{-}}

\newcommand{\pair}[2]{\cdabracks{#1,#2}}

\newcommand{\inl}[3]{\mathcd{inl}_{#2 \osum #3}{#1}}
\newcommand{\inr}[3]{\mathcd{inr}_{#2 \osum #3}{#1}}
\newcommand{\casearrow}{\Rightarrow}
\newcommand{\caseofflat}[7]{\ensuremath{\nmcase\,{#1}\,\mathcd{of}\,\mathcd{inl}\,\cdparens{#2\cdcolon #3}\,\casearrow{}\,{#4}\,\mathcd{|\,inr}\,\cdparens{#5\cdcolon #6}\,\casearrow\,{#7}\,\mathcd{end}}}

\renewcommand{\xi}{x_1}

\newcommand{\ai}{a_1}
\newcommand{\aii}{a_2}

\newcommand{\nletx}{\ensuremath{\mathcd{let*}}\xspace}
\newcommand{\nletbang}{\ensuremath{\mathcd{let!}}\xspace}
\newcommand{\nmcase}{\ensuremath{\mathcd{mcase}}\xspace}
\newcommand{\nreturn}{\ensuremath{\mathcd{return}}\xspace}

\newcommand{\letbin}[3]{\mathcd{let}\,\mathcd{!}\,{#1}\,\mathcd{be}\,{#2}\,\mathcd{in}\,{#3}\,\mathcd{end}}
\newcommand{\letpin}[4]{\mathcd{let}\,{#1}{\cross}{#2}\,\mathcd{be}\,{#3}\,\mathcd{in}\,{#4}\,\mathcd{end}}

\newcommand{\ifthenelse}[3]{\mathcd{if}\,{#1}\,\mathcd{then}\,{#2}\,\mathcd{else}\,{#3}}
\newcommand{\ite}[3]{\ifthenelse{#1}{#2}{#3}}

\newcommand{\ret}[1]{\mathcd{return}\cdparens{#1}}

\newcommand{\fun}[5]{\ensuremath{\mathcd{mfun}\,{#1}\,\cdparens{#2\cdcolon #3}\cdcolon{#4}\,\mathcd{is}\,{#5}\,\mathcd{end}}}

\newcommand{\funval}[6]{\ensuremath{\mathcd{mfun}_{\,#1}\,{#2}\cdparens{{#3}\cdcolon{#4}}\cdcolon{#5}\,\mathcd{is}\,{#6}\,\mathcd{end}}}

\newcommand{\oper}{o}
\newcommand{\op}[1]{o\cdparens{#1}}

\newcommand{\apply}[2]{\ensuremath{{#1}\,{#2}}}

\newcommand{\OpApply}[2]{\mathcd{app}\cdparens{#1\cdcomma #2}}

\newcommand{\bang}[1]{\mathcd{!}\,{#1}}

\newcommand{\bxed}[1]{{#1}\,\mathcd{box}}

\newcommand{\D}{\Delta}

\newcommand{\G}{\Gamma}

\newcommand{\cc}{{:}}

\newcommand{\ers}[1]{\ensuremath{{#1}^{-}}}

\title{\vspace{0.5in}Selective Memoization\thanks{This research was supported in part by NSF
    grants CCR-9706572, CCR-0085982, and CCR-0122581.}}

\author{ 
Umut A. Acar\footnote{\ttt{umut@mpi-sws.org} Max-Planck Institute for
  Software Systems.}
\and 
Guy E. Blelloch\footnote{\ttt{blelloch@cs.cmu.edu} Carnegie Mellon University. Pittsburgh, PA, USA.}
\and 
Robert Harper\footnote{\ttt{rwh@cs.cmu.edu} Carnegie Mellon University. Pittsburgh, PA, USA.}
}

\begin{document}
\maketitle
\vspace{0.5in}
\begin{abstract} 

  This paper presents language techniques for applying memoization
  selectively.  The techniques provide programmer control over
  equality, space usage, and identification of precise dependences so
  that memoization can be applied according to the needs of an
  application. Two key properties of the approach are that it accepts
  and efficient implementation and yields programs whose performance
  can be analyzed using standard analysis techniques.

  We describe our approach in the context of a functional language
  called \mfl and an implementation as a Standard ML library.  The
  \mfl language employs a modal type system to enable the programmer
  to express programs that reveal their true data dependences when
  executed.  We prove that the \mfl language is sound by showing that
  that \mfl programs yield the same result as they would with respect
  to a standard, non-memoizing semantics.  The SML implementation
  cannot support the modal type system of \mfl statically but instead
  employs run-time checks to ensure correct usage of primitives.

\end{abstract}

\section{Introduction} 
\label{sec:intro}

Memoization is a fundamental and powerful technique for result re-use.
It dates back a half century \citep{Bellman57, McCarthy63, Michie68}
and has been used extensively in many areas such as dynamic
programming~\citep{AhoHoUl74,Cohen83,CormenLeRi90,LiuSt99}, incremental
computation (e.g.,
~\citep{DemersReTe81,PughTe89,AbadiLaLe96,LiuStTe98,HeydonLeYu00}), and
others~\citep{Bird80,MostowCo85,Hughes85,Norvig91,LiuStTe98}.  In fact,
lazy evaluation provides a limited form of
memoization~\citep{Peyton-Jones87}.

Although memoization can dramatically improve performance and can
require only small changes to the code, no language or library support
for memoization has gained broad acceptance.  Instead, many successful
uses of memoization rely on application-specific support code.  The
underlying reason for this is one of control: since memoization is all
about performance, the user must be able to control the performance of
memoization.  Many subtleties of memoization, including the cost of
equality checking and the cache replacement policy for memo tables,
can make the difference between exponential and linear running time.

To be general and widely applicable a memoization framework must
provide control over these three areas: (1) the kind and cost of
equality tests; (2) the identification of precise dependences between
the input and the output of memoized code; and (3) space
management. Control over equality tests is critical, because this is
how re-usable results are identified.  Identifying precise dependences
is important to maximize result reuse.  Being able to control when
memo tables or individual entries are purged is critical, because
otherwise the user will not know whether or when results are re-used.

In this paper, we propose techniques for memoization that provide
control over equality and identification of dependences, and some
control over space management.  We study the techniques in the context
of a small language called \mfl, which is a purely functional language
enriched with support for user-controlled, selective memoization.  We
give several examples of the use of the language and we prove its type
safety and correctness---{\em i.e.,} that the semantics are preserved
with respect to a non-memoized version.  The operational semantics of
\mfl specifies the performance of programs accurately enough to
determine (expected) asymptotic time bounds.\footnote{Expected, rather
  than worst-case, performance is required because of our reliance on
  hashing.}  As an example, we show how to analyze the performance of
a memoized version of Quicksort.  The \mfl language accepts an
efficient implementation with expected constant-time overhead by
representing memo tables with nested hash tables.  We give an
implementation of \mfl as a library for the Standard ML language.  The
implementation cannot support the modal type system of \mfl
statically; instead, it relies on run-time checks to ensure correct
usage of memoization primitives.

In the next section we describe background and related work.  In
\secref{framework} we introduce our approach via some examples.  In
\secref{mfl} we formalize the \MFL{} language and discuss its safety,
correctness, and performance properties.  In \secref{imp} we present a
simple implementation of the framework as a Standard ML library.  In
\secref{discussion} we discuss several ways in which the approach may
be extended.  

This paper extends the conference version~\citep{AcarBlHa03} with the
proofs for the correctness of the proposed approach and with a more
detailed description of the implementation.  Although the
implementation provided here is in the form of a simple library, some
of the techniques proposed here have been implemented in
CEAL~\citep{HammerAcCh09} and Delta ML~\citep{Ley-WildFlAc08,AcarLW09}
languages that provide direct support for self-adjusting computation.

\section{Background and Related Work}
\label{sec:background}

A typical memoization scheme maintains a memo table mapping argument
values to previously computed results.  This table is consulted before
each function call to determine if the particular argument is in the
table.  If so, the call is skipped and the result is returned;
otherwise the call is performed and its result is added to the table.
The semantics and implementation of the memo lookup are critical to
performance.  Here we review some key issues in implementing
memoization efficiently.

\subsection{Equality}
Any memoization scheme needs to search a memo table for a match to the
current arguments.  Such a search, at minimum, requires a test for
equality.  Typically it also requires some form of hashing.  In
standard language implementations testing for equality on structures,
for example, can require traversing the whole structure.  The cost of
such an equality test can negate the advantage of memoizing and may
even change the asymptotic behavior of the function.  A few approaches
have been proposed to alleviate this problem.  The first is based on
the fact that for memoization equality need not be exact---it can
return unequal when two arguments are actually equal.  The
implementation could therefore decide to skip the test if the equality
is too expensive, or could use a conservative equality test, such as
``location'' equality.  The problem with such approaches is that
whether a match is found could depend on particulars of the
implementation and will surely not be evident to the programmer.

Another approach for reducing the cost of equality tests is to ensure
that there is only one copy of every value, via a technique known as
``hash consing''~\citep{GotoKa76,Allen78,SpitzenLe78}.  If there is
only one copy, then equality can be implemented by comparing
locations.  In fact, the location can also be used as a key to a hash
table.  In theory, the overhead of hash-consing is constant in the
expected case (expectation is over internal randomization of hash
functions).  In practice, hash-consing can be expensive because of
large memory demands and interaction with garbage collection.  In
fact, several researchers have argued that hash-consing is too
expensive for practical
purposes~\citep{Pugh88,AppelGo93,MurphyHaCr02}.  As an
alternative to hash consing, Pugh proposed lazy structure
sharing~\citep{Pugh88}.  In lazy structure sharing whenever two equal
values are compared, they are made to point to the same copy to speed
up subsequent comparisons.  As Pugh points out, the disadvantage of
this approach is that the performance depends on the order comparisons
and can therefore be difficult to analyze.

We note that even with hash-consing, or any other method, it remains
critical to define equality on all types including reals and
functions.  Claiming that functions are never equivalent, for example,
is not satisfactory because the result of a call involving some
function as a parameter will never be re-used.

\subsection{Precise dependences}
To maximize result re-use, the result of a function call must be
stored with respect to its true dependences.  This issue arises when
the function examines only parts or an approximation of its parameter.
To enable ``partial'' equality checks, the unexamined parts of the
parameter should be disregarded.  To increase result re-use, the
programmer should be able to match on the approximation, rather than
the parameter itself.  As an example, consider the code
\begin{verbatim}
  fun f(x,y,z) = if (x > 0) then fy(y) else fz(z)
\end{verbatim}
\noindent
The result of {\tt f} depends on either {\tt (x,y)} or {\tt (x,z)}.
Also, it depends on an approximation of {\tt x} (whether or not it is
positive) rather than its exact value.  For example, the memo entry
for \ttt{f(7,11,20)} should match the calls \ttt{f(7,11,30)} and
\ttt{f(4,11,50)}, since when {\tt x} is positive, the result depends
only on {\tt y}.

Several researchers have remarked that partial matching can be very
important in some
applications~\citep{PenningsSwVo92,Pennings94,AbadiLaLe96,HeydonLeYu00}.
Abadi, Lampson, L\'evy~\citep{AbadiLaLe96}, and Heydon, Levin,
Yu~\citep{HeydonLeYu00} have suggested program analysis methods for
tracking dependences for this purpose.  Although their technique is
likely effective in catching potential matches, it does not provide a
programmer controlled mechanism for specifying what dependences should
be tracked.  Also, their program analysis technique can change the
asymptotic complexity of a program, making it difficult to asses the
effects of memoization.

\subsection{Space management} 
Another problem with memoization is its space requirement.  As a
program executes, its memo tables can become large and limit the
utility of memoization.  To alleviate this problem, memo tables or
individual entries should be disposed of under programmer control.

In some applications, such as in dynamic programming, most result
re-use occurs among the recursive calls of some function.  Thus, the
memo table of such a function can be disposed of whenever it
terminates.  This can be achieved by associating a memo table with a
each memoized function and reclaiming the table when the function goes
out of scope~\citep{CookLa97,Hughes85}.

In other applications, where result re-use is less structured,
individual memo table entries should be purged according to a
replacement policy~\citep{Hilden76,Pugh88}.  The problem is to
determine what exact replacement policy should be used and to analyze
the performance effects of the chosen policy.  One widely used
approach is to replace the least recently used entry.  Other, more
sophisticated, policies have also been suggested~\citep{Pugh88}.  In
general the replacement policy must be application-specific, because,
for any fixed policy, there are programs whose performance is made
worse by that choice~\citep{Pugh88}.

\subsection{Memoization and dynamic dependence graphs}
The techniques presented in this paper were motivated by our previous
work on adaptive computation~\citep{AcarBlHa02}.  In subsequent
work~\citep{Acar05,AcarBlBlHaTa09,AcarBlDo07}, we showed that
memoization and adaptive computation techniques are duals in the way
that they provide for computation re-use.  Based on this duality, we
showed that they can be combined to provide an incremental-computation
technique, called self-adjusting computation, that achieves efficiency
for a reasonably broad range of applications.  Perhaps one of the most
interesting aspects of this combination is that it enables re-use of
computations under mutations to
memory~\citep{AcarBlBlHaTa06,AcarBlDo07}, which turns out to be
critically important for effective re-use of computations. This later
work subsequently has led to the development of
CEAL~\citep{HammerAcCh09} and Delta ML
languages~\citep{Ley-WildFlAc08,AcarLW09} for self-adjusting
computation.  These languages provide language constructs to enable
the user to memoize expression as needed and support creation of
locally scoped memo tables.  Delta ML additionally supports
user-provided equality tests.

\section{A Framework for Selective Memoization}
\label{sec:framework}

We present an overview of our approach via several examples.  The
examples are written in an language that extends a purely functional,
ML-like language with selective-memoization primitives. We formalize
the core of this language and study its safety, soundness, and
performance properties in~\secref{language}.

\subsection{Incremental exploration with resources}
Our approach enables the programmer to determine the precise
dependences between the input and the result of a function.  The main
idea is to deem the parameters of a function as {\em resources} and
provide primitives to explore incrementally any value, including the
underlying value of a resource.  This incremental exploration process
reveals the dependences between the parameter of the function and its
result.  

The incremental exploration process is guided by types.  If a
value has the modal type $\bang{\tau}$, then the underlying value of
type $\tau$ can be bound to an ordinary, unrestricted variable by the
\nletbang construct; this will create a dependence between the
underlying value and the result.  If a value has a product type, then
its two parts can be bound to two resources using the \nletx
construct; this creates no dependences.  If the value is a sum type,
then it can be case analyzed using the \nmcase construct, which
branches according to the outermost form of the value and assigns the
inner value to a resource; {\tt mcase} creates a dependence on the
outer form of the value of the resource.  The key aspect of the \nletx
and \nmcase is that they bind resources rather than ordinary
variables.

\begin{figure}
\small
\begin{center}
\setlength{\tabcolsep}{\codecolsep}
\small
\begin{tabular}{|l|l|}
\hline
Non-memoized 
& 
Memoized\\
\hline
\parbox[t]{1in}{
\begin{codeListing}
XX\=XX\=XX\=XX\=XX\=XXXXXXXXXXXXXXXXXXXXXXXXXXXXXXXXXXXXXXXXXXXXXXX\=\kill
fib:int -> int\\
\kfun fib (n)=\\
\>\\
\>\kif (n < 2) \kthen n\\
\>\kelse fib(n-1) + fib(n-2)
\end{codeListing}}
& 
\parbox[t]{1in}{
\begin{codeListing}
XX\=XX\=XX\=XX\=XX\=XXXXXXXXXXXXXXXXXXXXXXXXXXXXXXXXXXXXXXXXXXXXXXX\=\kill
mfib:!int -> int\\
\kmfun mfib (n')= \\
\>\klet~ !n = n' \kin \kreturn (\\
\>\>\kif (n < 2) \kthen n \\
\>\>\kelse mfib(!(n-1)) + mfib(!(n-2)))
\kend
\end{codeListing}}\\
\hline
\parbox[t]{2in}{
\begin{codeListing}
XX\=XX\=XX\=XX\=XX\=XXXXXXXXXXXXXXXXXXXXXXXXXXXXXXXXXXXXXXXXXXXXXXX\=\kill
f: int * int * int -> int\\
\kfun f (x, y, z)=\\
\>\kif (x > 0) \kthen \\
\>\>fy y\\
\>\kelse\\
\>\>fz z
\end{codeListing}}
&
\parbox[t]{2in}{
\begin{codeListing}
XX\=XX\=XX\=XX\=XX\=XXXXXXXXXXXXXXXXXXXXXXXXXXXXXXXXXXXXXXXXXXXXXXX\=\kill
mf:int * !int * !int -> int\\
\kmfun mf (x', y', z')=\\
\>\kmif (x' > 0) \kthen \\
\>\>\klet~!y = y' \kin \kreturn (fy y) \kend\\
\>\kelse\\
\>\>\klet~!z = z' \kin \kreturn (fz z) \kend
\end{codeListing}}\\
\hline
\end{tabular}
\end{center}
\vspace*{-2mm}
\caption{Fibonacci and expressing partial dependences.}
\label{fig:precise-dependences}
\end{figure}

Exploring the input to a function via \nletbang, \nmcase, and \nletx
builds a {\em branch} recording the dependences between the input and
the result of the function. The \nletbang adds to the branch the full
value, the \nmcase adds the kind of the sum, and \nletx adds nothing.
Consequently, a branch contains both data dependences (from
\nletbang's) and control dependences (from \nmcase's).  When a
\nreturn is encountered, the branch recording the revealed dependences
is used to key the memo table.  If the result is found in the memo
table, then the stored value is returned, otherwise the body of the
\nreturn is evaluated and the memo table is updated to map the branch
to the result.  The type system ensures that all dependences are made
explicit by precluding the use of resources within \nreturn's body.

As an example consider the Fibonacci function {\tt fib} and its
memoized counterpart {\tt mfib} shown in~\figref{precise-dependences}.
The memoized version, {\tt mfib}, exposes the underlying value of its
parameter, a resource, before performing the two recursive calls as
usual.  Since the result depends on the full value of the parameter,
it has a bang type.  The memoized Fibonacci function runs in linear
time as opposed to exponential time when not memoized.

Partial dependences between the input and the result of a function can be
captured by using the incremental exploration technique.  As an example
consider the function {\tt f} shown in
\figref{precise-dependences}.  The function checks whether {\tt x} is
positive or not and returns {\tt fy(y)} or {\tt fz(z)}.  Thus the result
of the function depends on an approximation of {\tt x} (its sign) and on
either {\tt y} or {\tt z}.  The memoized version {\tt mf} captures this by
first checking if {\tt x'} is positive or not and then exposing the
underlying value of {\tt y'} or {\tt z'} accordingly.  Consequently, the
result will depend on the sign of {\tt x'} and on either {\tt y'} or {\tt
z'}.  Thus if {\tt mf} is called with parameters $(1,5,7)$ first and then
$(2,5,3)$, the result will be found in the memo the second time, because
when {\tt x'} is positive the result depends only on {\tt y'}.  Note that
{\tt mif} construct used in this example is just a special case of the
more general {\tt mcase} construct.

\subsection{Memo lookups and indexable types}

A critical issue for efficient memoization is the implementation of
memo tables along with lookup and update operations on them.  We
support expected constant time memo-table lookup and update operations
by representing memo tables using hashing. This requires that the
underlying type $\tau$ of a modal type {\tt !$\tau$} be an {\em
  indexable type}.  An indexable type is associated with an injective
function, called an {\em index function}, that maps each value of that
type to a unique integer called the {\em index}.  The uniqueness
property of the indices for a given type ensures that two values are
equal if and only if their indices are equal. We define equality only
for indexable types.  This enables implementing memo tables as hash
tables keyed by branches consisting of indices.

We assume that each primitive type comes with an index function.  For
example, for integers, the identity function can be chosen as the
index function.  Composite types such as lists or functions must be
{\em boxed} to obtain an indexable type.  A boxed value of type $\tau$
has type $\bxed{\tau}$.  When a box is created, it is assigned a
unique tag, and this tag is used as the unique index of that boxed
value.  For example, we can define boxed lists as
follows.\\[1mm]
\parbox{2in}{
\begin{codeListing}
XX\=XX\=XX\=XX\=XX\=XXXXXXXXXXXXXXXXXXXXXXXXXXXXXXXXXXXXXXXXXXXXXXX\=\kill
datatype $\alpha$ blist' =~~NIL~|~CONS of $\alpha$ * (($\alpha$ blist') box)\\ 
type $\alpha$ blist = ($\alpha$ blist') box
\end{codeListing}
}

\smallskip

Based on boxes we implement hash-consing as a form of memoization.
For example, hash-consing for boxed lists can be implemented as follows.\\[1mm]
\parbox{2in}{
\begin{codeListing}
XX\=XX\=XX\=XX\=XX\=XXXXXXXXXXXXXXXXXXXXXXXXXXXXXXXXXXXXXXXXXXXXXXX\=\kill
hCons: !$\alpha$ * !($\alpha$ blist) -> $\alpha$ blist \\
\kmfun hCons (h', t') = \\[0.5ex]
\>\klet~!h = h' \kand~!t = t' \kin\\
\>\>return (box (CONS(h,t))) \\
\>\kend
\end{codeListing}}\\
The function takes an item and a boxed list and returns the boxed list
formed by consing them.  Since the function is memoized, if it is ever
called with two values that are already hash-consed, then the same
result will be returned.  The advantage of being able to define
hash-consing as a memoized function is that it can be applied
selectively.

\begin{figure}[t]
\begin{center}
\setlength{\tabcolsep}{5pt}
\small
\begin{tabular}{|l|l|}
\hline
Non-memoized 
& 
Memoized\\
\hline
\parbox[t]{2in}{
\begin{codeListing}
XX\=XX\=XX\=XX\=XX\=XX\=XpXXXXXXXXXXXXXXXXXXXXXXXXXXXXXXXXXXXXXXXXXXX\=\kill
ks: int * ((int*real) list) -> int\\
\kfun ks (c,l) = \\
\>\\
\>\kcase l \kof \\
\>\>nil => 0\\
\>~|(w,v)::t =>\\
\>\>\kif (c < w) \kthen\\
\>\>\>ks(c,t)\\
\>\>\kelse \\
\>\>\>\klet
          v1 = ks(c,t)\\
\>\>\>\>\>v2 = v + ks(c-w,t)\\
\>\>\>\kin\\
\>\>\>\>\kif (v1>v2) then v1\\
\>\>\>\>\kelse v2\\
\>\>\>\kend
\end{codeListing}}
&
\parbox[t]{2in}{
\begin{codeListing}
XX\=XX\=XX\=XX\=XX\=XX\=XX\=XXXXXXXXXXXXXXXXXXXXXXXXXXXXXXXXXXXXXXXXXXX\=\kill
mks: !int * !((int*real) list) -> int\\
\kmfun mks (c',l')\\
\>\klet~!c = c' \kand~!l = l' \kin \kreturn (\\
\>\>\kcase (unbox l) \kof \\
\>\>\>\kNIL => 0\\
\>\>|~\kCONS((w,v),t) => \\
\>\>\>\kif (c < w) \kthen \\
\>\>\>\>mks(!c,!t)\\
\>\>\>\kelse\\
\>\>\>\>\klet
            v1 = mks(!c,!t)\\
\>\>\>\>\>\>v2 = v + mks(!(c-w),!t)\\
\>\>\>\>\kin\\
\>\>\>\>\>\kif (v1 > v2) then v1\\
\>\>\>\>\>\kelse v2\\
\>\>\>\>\kend) \kend
\end{codeListing}}\\
\hline
\end{tabular}
\end{center}
\vspace*{-2mm}
\caption{Memo tables for memoized Knapsack can be discarded at
completion.}
\label{fig:ks}
\end{figure}

\subsection{Controlling space usage via scoping}
To control space usage of memo tables, we enable the programmer to
dispose of memo tables by conventional scoping by assoaciating each
memoized function with its own memo table.  When a memoized function
goes out of scope, its memo table can be garbage collected.  For
example, in many dynamic-programming algorithms result re-use occurs
between recursive calls of the same function.  In this case, the
programmer can scope the memoized function inside an auxiliary
function so that its memo table is discarded as soon as the auxiliary
function returns.  As an example, consider the standard algorithm for
the Knapsack Problem {\tt ks} and its memoized version {\tt
  mks}~\figref{ks}.  Since result sharing mostly occurs among the
recursive calls of {\tt mks}, it can be scoped in some other function
that calls {\tt mks}; once {\tt mks} returns its memo table will go
out of scope and can be discarded.  

We note that this technique gives only partial control over space
usage.  In particular it does not give control over when individual
memo table entries are purged.  In \secref{discussion}, we discuss how
the framework might be extended so that each memo table is managed
according to a programmer specified caching scheme.  The main idea is
to require the programmer to supply a caching scheme as a parameter to
the {\tt mfun} and maintain the memo table according to the chosen
caching scheme.

\begin{figure}[t]
\begin{center}
\small
\begin{tabular}{|l|l|}
\hline
Non-memoized 
& 
Memoized\\
\hline
\parbox[t]{2in}{
\begin{codeListing}
XX\=XX\=XX\=XX\=XX\=XX\=XXXXXXXXXXXXXXXXXXXXXXXXXXXXXXXXXXXXXXXXXXXXX\=\kill
\>\\
\>\\
fil: int->bool * int list -> int list\\
\kfun fil (g:int->bool, l:int list) =\\
\>\kcase l \kof \\
\>\>\knil => \knil\\
\>|~h::t =>\\
\>\>\klet tt = fil(g,t) \kin\\
\>\>\>\kif (g h) \kthen h::tt\\
\>\>\>\kelse tt\\
\>\>\kend
\end{codeListing}}
&
\parbox[t]{2in}{
\begin{codeListing}
XX\=XX\=XX\=XX\=XX\=XX\=XXXXXXXXXXXXXXXXXXXXXXXXXXXXXXXXXXXXXXXXXXXXX\=\kill
empty = box NIL\\
\>\\
mfil: int->bool * int blist -> int blist\\
\kfun mfil (g,l) = \\
\>\kcase (unbox l) \kof \\
\>\>\kNIL => empty\\
\>|~\kCONS(h,t) =>  \\
\>\>\klet tt = mfil(g,t) \kin\\
\>\>\>\kif (g h) \kthen hCons(h,tt)\\
\>\>\>\kelse tt\\
\>\>\kend
\end{codeListing}}\\[-4mm]
\parbox[t]{2in}{
\begin{codeListing}
XX\=XX\=XX\=XX\=XX\=XXXXXXXXXXXXXXXXXXXXXXXXXXXXXXXXXXXXXXXXXXXXXXX\=\kill
qs: int list -> int list\\
\kfun qs (l) = \\
\>\\
\>\kcase l \kof \\
\>\>\knil => nil\\
\>|~\kcons(h,t) =>  \\
\>\>\klet
        s = fil(fn x=>x<h,t)\\
\>\>\>\>g = fil(fn x=>x>=h,t)\\
\>\>\kin\\
\>\>\>(qs s)@(h::(qs g))\\
\>\>\kend
\end{codeListing}}
&
\parbox[t]{2in}{
\begin{codeListing}
XX\=XX\=XX\=XX\=XX\=XX\=XXXXXXXXXXXXXXXXXXXXXXXXXXXXXXXXXXXXXXXXXXXXX\=\kill
mqs: !(int blist) -> int blist\\
\kmfun mqs (l':!int blist) = \\
\>\klet~!l = l' \kin \kreturn ( \\
\>\>\kcase (unbox l) \kof \\
\>\>\>\kNIL => NIL\\
\>\>|~\kCONS(h,t) =>  \\
\>\>\>\klet
          s = mfil(fn x=>x<h,t)\\
\>\>\>\>\>g = mfil(fn x=>x>=h,t)\\
\>\>\>\kin\\
\>\>\>\>(mqs !s)@(h::(mqs !g))\\
\>\>\>\kend) \kend
\end{codeListing}}\\
\hline
\end{tabular}
\end{center}
\vspace*{-2mm}
\caption{The Quicksort algorithm.}
\label{fig:qsort}
\end{figure}

\subsection{Memoized Quicksort}

As a more sophisticated example, we consider Quicksort. \figref{qsort}
shows an implementation of the Quicksort algorithm and its memoized
counterpart. The algorithm first divides its input into two lists
containing the keys less than the pivot, and greater than the pivot by
using the filter function {\tt fil}.  It then sorts the two sublists,
and returns the concatenation of the results.  The memoized filter
function {\tt mfil} uses hash-consing to ensure that there is only one
copy of each result list.  The memoized Quicksort algorithm {\tt mqs}
exposes the underlying value of its parameter and is otherwise similar
to {\tt qs}.  Note that {\tt mqs} does not build its result via
hash-consing---it can output two copies of the same result.  Since in
this example the output of {\tt mqs} is not consumed by any other
function, there is no need to do so.  Even if the result were consumed
by some other function, one can choose not to use hash-consing
because operations such as insertions to and deletions from the input
list will surely change the result of Quicksort.

When the memoized Quicksort algorithm is called on ``similar'' inputs,
one would expect that some of the results would be re-used.  Indeed,
we show that the memoized Quicksort algorithm computes its result in
expected linear time when its input is obtained from a previous input
by inserting a new key at the beginning.  Here the expectation is over
all permutations of the input list and also the internal randomization
of the hash functions used to implement the memo tables.  For the
analysis, we assume, without loss of generality, that all keys in the
list are unique.

\begin{theorem}
Let $L$ be a list and let $L'=[a,L]$.  Consider running memoized Quicksort
on $L$ and then on $L'$.  The running time of Quicksort on the modified
list $L'$ is expected $O(n)$ where $n$ is the length of $L'$.
\end{theorem}
\begin{proof} 
Consider the recursion tree of Quicksort with input $L$, denoted $Q(L)$,
and label each node with the pivot of the corresponding recursive call
(see \figref{pivots} for an example).  Consider any pivot (key) $p$ from
$L$ and let $L_p$ denote the keys that precede $p$ in $L$.  It is easy to
see that a key $k$ is in the subtree rooted at $p$ if and only if the
following two properties are satisfied for any key $k' \in L_p$.
\begin{enumerate}
\item If $k' < p$ then $k > k'$, and
\item if $k' > p$ then $k < k'$.
\end{enumerate}
Of the keys that are in the subtree of $p$, those that are less than
$p$ are in its left subtree and those greater than $p$ are in its
right subtree.

\begin{figure}
\centering{\parbox{5in}{
\includegraphics[width=2in,height=2in]{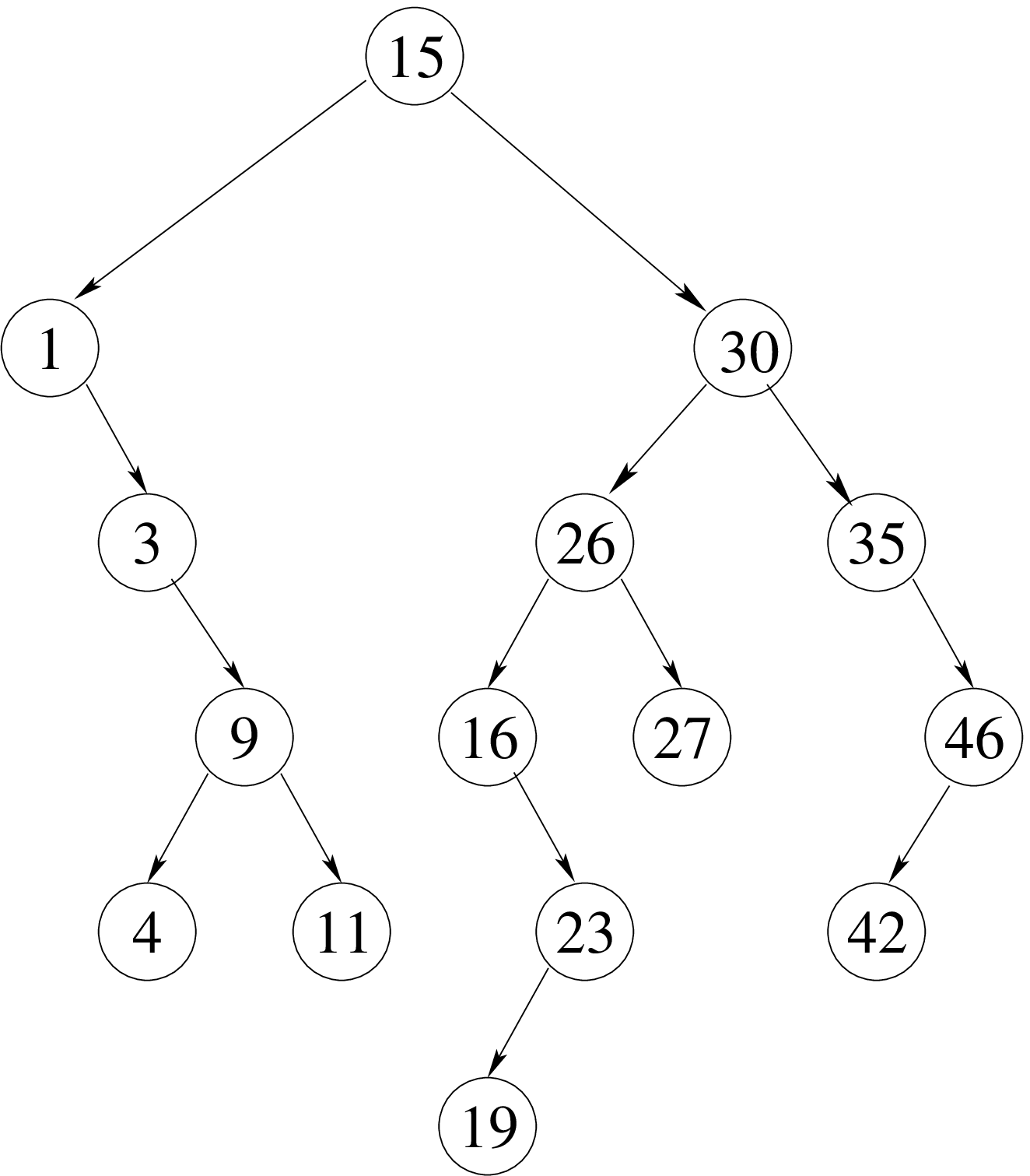} \hfill
\includegraphics[width=2in,height=2in]{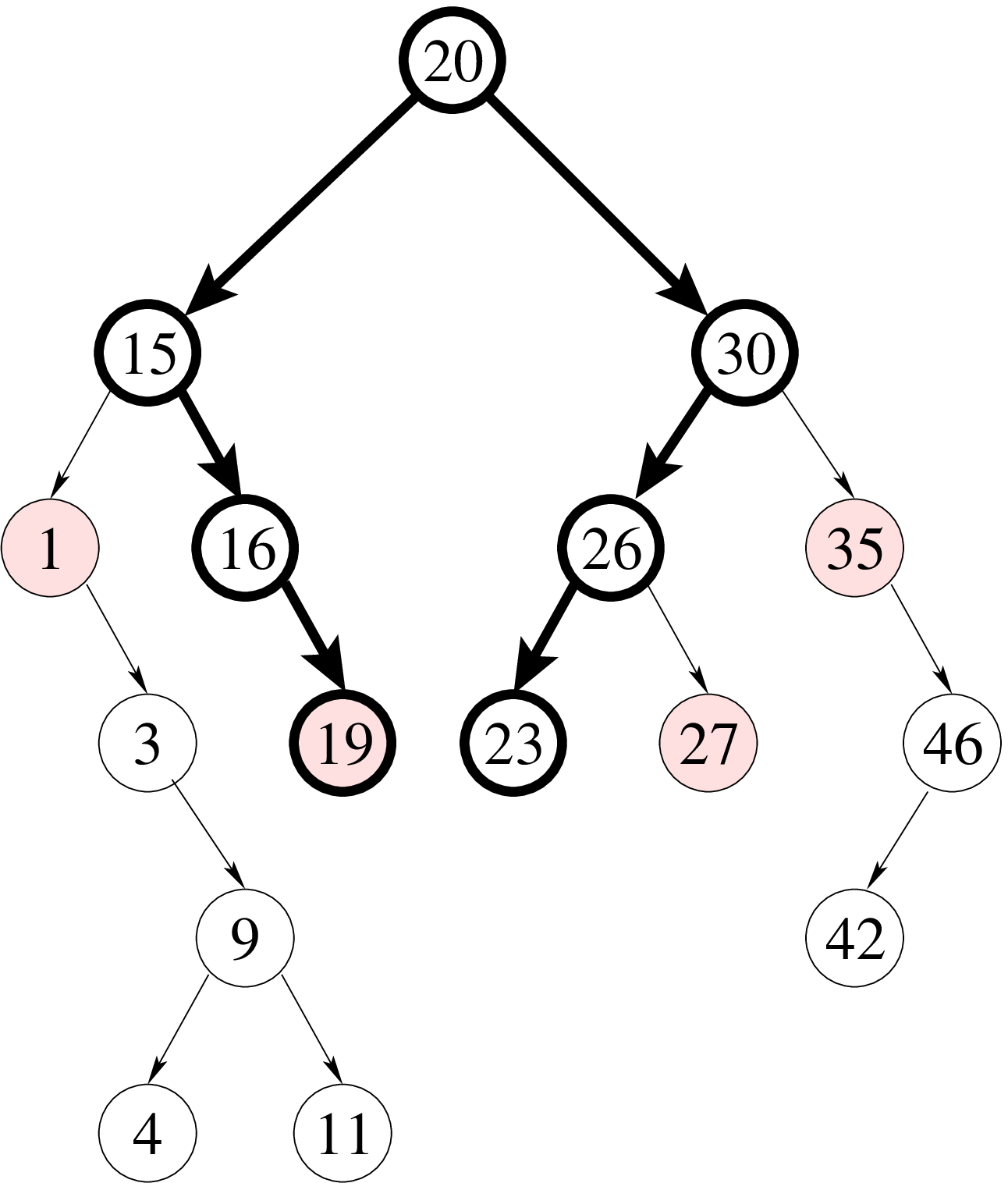} \hfill
}}
\vspace*{-2mm}
\caption{The recursion tree for Quicksort with inputs  $L =
[15,30,26,1,3,16,27,9,35,4,46,23,11,42,19]$ (left) and \mbox{$L' =
[20,L]$ (right).}}
\label{fig:pivots}
\end{figure}

Now consider the recursion tree $Q(L')$ for $L'= [a,L]$ and let $p$ be
any pivot in $Q(L')$.  Suppose $p < a$ and let $k$ be any key in the
left subtree of $p$ in $Q(L)$.  Since $k < p$, by the two properties
$k$ is in the left subtree of $p$ in $Q(L')$.  Similarly if $p> a$
then any $k$ in the right subtree of $p$ in $Q(L)$ is also in the
right subtree of $p$ in $Q(L')$.  Since filtering preserves the
respective order of keys in the input list, for any $p$, $p<a$, the
input to the recursive call corresponding to its left child will be
the same.  Similarly, for $p>a$, the input to the recursive call
corresponding to its right child will be the same.  Thus, when sorting
$L'$ these recursive calls will find their results in the memo.
Therefore only recursive calls corresponding to the root, to the
children of the nodes in the rightmost spine of the left subtree of
the root, and the children of the nodes in the leftmost spine of the
right subtree of the root may be executed (the two spines are shown
with thick lines in \figref{pivots}).  Furthermore, the results for
the calls adjacent to the spines will be found in the memo.

Consider the calls whose results are not found in the memo.  In the
worst case, these will be all the calls along the two spines.
Consider the sizes of inputs for the nodes on a spine and define the
random variables $X_1 \ldots X_k$ such that $X_i$ is the least number
of recursive calls (nodes) performed for the input size to become
$\left(\frac{3}{4}\right)^in$ or less after it first becomes
$\left(\frac{3}{4}\right)^{(i-1)}n$ or less. Since $k \le
\lceil\log_{4/3}{n}\rceil$, the total and the expected number of operations along a
spine are
\begin{eqnarray*}
C(n) & \le & \sum_{i=1}^{\lceil\log_{4/3}{n}\rceil}{X_i \left(\frac{3}{4} \right) ^{i-1}n},\mbox{~and}\\
E[C(n)] &\le & \sum_{i=1}^{\lceil\log_{4/3}{n}\rceil}{E[X_i]\left(\frac{3}{4}\right)^{i-1}n}.
\end{eqnarray*}\\[1ex]

Since the probability that the pivot lies in the middle half of
the list is $\frac{1}{2}$, $E[X_i] \le 2$ for $i\ge1$, and we have 
\begin{eqnarray*}
E[C(n)] & \le & \sum_{i=1}^{\lceil\log_{4/3}{n}\rceil}{2\left(\frac{3}{4}\right)^{i-1} n}.
\end{eqnarray*}
Thus, $E[C(n)] = O(n)$ This bound holds for both spines; therefore the
number of operations due to calls whose results are not found in the
memo is $O(n)$.  Since each operation, including hash-consing, takes
expected constant time, the total time of the calls whose results are
not in the memo is $O(n)$. Now, consider the calls whose results are
found in the memo, each such call will be on a spine or adjacent to
it, thus there are an expected $O(\log{n})$ such calls.  Since, the
memo table lookup overhead is expected constant time the total cost
for these is $O(\log{n})$. We conclude that Quicksort will take
expected $O(n)$ time for sorting the modified list~$L'$.
\end{proof}

This theorem can be extended to show that the $O(n)$ bound holds for
an insertion anywhere in the list.  Although this bound is better than
a complete rerun, which would take expected $O(n\log{n})$, it is still
far from optimal for Quicksort (expected $O(\log{n})$).  It is not
known if the optimal bound can be achieved by using memoization only.
The optimal, however, can be achieved by using a combination of
dynamic dependence graphs and memoization~\citep{AcarBlBlTa06,Acar05}.

\section{The \MFL{} Language}
\label{sec:mfl}
\label{sec:language}

In this section we study a small functional language, called \MFL{},
that supports selective memoization.  \MFL{} distinguishes memoized
from non-memoized code, and is equipped with a modality for tracking
dependences on data structures within memoized code.  This modality is
central to our approach to selective memoization, and is the focus of
our attention here.  The main result is a soundness theorem stating
that memoization does not affect the outcome of a computation compared
to a standard, non-memoizing semantics (\secref{mfl::soundness}).  We
also show that the memoization mechanism of \MFL{} causes a constant
factor slowdown compared to a standard, non-memoizing semantics
(\secref{mfl::complexity}).

\subsection{Abstract syntax}
\label{sec:mfl-abs-syn}
The abstract syntax of \MFL{} is given in Figure~\ref{fig:ast}.  The
meta-variables $x$ and $y$ range over a countable set of
\emph{variables}.  The meta-variables $a$ and $b$ range over a
countable set of \emph{resources}.  (The distinction will be made
clear below.)  The meta-variable $l$ ranges over a countable set of
\emph{locations}.  We assume that variables, resources, and locations
are mutually disjoint.  The binding and scope conventions for
variables and resources are as would be expected from the syntactic
forms. As usual we identify pieces of syntax that differ only in their
choice of bound variable or resource names.  A term or expression is
\emph{resource-free} if and only if it contains no free resources, and
is \emph{variable-free} if and only if it contains no free variables.
A \emph{closed} term or expression is both resource-free and
variable-free; otherwise it is \emph{open}.

\begin{figure}
\small
\[
\renewcommand{\arraycolsep}{2mm}
\renewcommand{\arraystretch}{1}
\begin{array}{|@{~~~~}lrcl@{~~~~}|}
\hline
&~&~& \\
\textit{Indexable Types} & \eta & \bnfdef &  \tunit \bnfalt \Int \bnfalt \ldots\\[2mm]

\textit{Types}  & \tau & \bnfdef & \eta \bnfalt \bang{\eta} \bnfalt 
\taui \cross \tauii \bnfalt \taui + \tauii \bnfalt \rt{\u}{\tau} \bnfalt
\taui \ra \tauii \\[2ex]

\textit{Operators} & o   & \bnfdef       &  \Plus \bnfalt \Minus \bnfalt \ldots \\[2mm]

\textit{Expressions} & \e     & \bnfdef     &
         \ret{\t} \bnfalt \letbin{x \cdcolon \eta}{\t}{\e}\bnfalt \\
 & & & \letpin{a_1\cdcolon \taui}{a_2\cdcolon \tauii}{\t}{\e} \bnfalt \\
 & & &  \caseofflat{\t}{\ai}{\taui}{\ei}{\aii}{\tauii}{\eii}\\[2mm]

\textit{Terms} & \t     & \bnfdef     & v \bnfalt 
        \op{\ti,\ldots,\t_n} \bnfalt \pair{\ti}{\tii} \bnfalt \fun{f}{a}{\taui}{\tauii}{\e} \bnfalt \\
& & &   \apply{\ti}{\tii} \bnfalt \bang{\t} \bnfalt
        \inl{\t}{\taui}{\tauii} \bnfalt \inr{\t}{\taui}{\tauii} \bnfalt   \roll{\t} \bnfalt \unroll{\t} \\[2mm]

\textit{Values} & v        & \bnfdef    & x \bnfalt a \bnfalt  \uno \bnfalt n \bnfalt
\bang{v} \bnfalt \pair{\vi}{\vii} \bnfalt  \funval{l}{f}{a}{\taui}{\tauii}{\e}\\

&~&~& \\
\hline
\end{array}
\]

\vspace{-2mm}
\caption{The abstract syntax of \MFL{}.}
\label{fig:ast}
\end{figure}

The types of \mfl include \textcd{1} (unit), \textcd{int}, products
and sums, recursive data types $\rt{\u}{\tau}$, memoized function
types, and bang types $\bang{\eta}$.  \mfl distinguishes {\em
  indexable types}, denoted $\eta$, as those that accept an injective
function, called an {\em index function}, whose co-domain is
integers. The underlying type of a bang type $\bang{\eta}$ is
restricted to be an indexable type.  For type \textcd{int}, identity
serves as an index function; for \textcd{1} (unit) any constant
function can be chosen as the index function.  For non-primitive types
an index can be supplied by boxing values of these types.  Boxed
values would be allocated in a store and the unique location of a box
would serve as an index for the underlying value.  With this extension
the indexable types would be defined as \mbox{$\eta \bnfdef \tunit
  \bnfalt \Int \bnfalt \tau~\mathcd{box}$}.  Since supporting boxed
types is well understood, we do not formalize boxing here.

The abstract syntax is structured into \emph{terms} and
\emph{expressions}, in the terminology of Pfenning and
Davies~\citep{pfenning+:judgmental}.  Roughly speaking, terms evaluate
independently of their context, as in ordinary functional programming,
whereas expressions evaluate in the context of a memo table.  Thus,
the body of a memoized function is an expression, whereas the function
itself is a term.  Note, however, that the application of a function
is a \emph{term}, not an \emph{expression}; this corresponds to the
encapsulation of memoization with the function, so that updating the
memo table is benign.  In a more complete language we would include
case analysis and projection forms among the terms, but for the sake
of simplicity we include these only as expressions.  We would also
include a plain function for which the body is a term.  Note that
every term is trivially an expression; the \textcd{return} expression
is the inclusion.

\subsection{Static semantics}
\label{sec:mfl-static}

The type structure of \MFL{} extends the framework of Pfenning and
Davies ~\citep{pfenning+:judgmental} with a ``necessitation'' modality,
$\bang{\eta}$, which is used to track data dependences for selective
memoization.  This modality does not correspond to a monadic
interpretation of memoization effects ($\circle\,\tau$ in the notation
of Pfenning and Davies), though one could imagine adding such a
modality to the language.  The introductory and eliminatory forms for
necessity are standard, namely $\bang{\t}$ for introduction, and
$\letbin{x\cdcolon\eta}{\t}{\e}$ for elimination.

Our modality demands that we distinguish variables from resources.
Variables in \MFL{} correspond to the ``validity'', or
``unrestricted'', context in modal logic, whereas resources in \MFL{}
correspond to the ``truth'', or ``restricted'' context.  An analogy
may also be made to the judgmental presentation of linear
logic~\citep{Pfenning95,PolakowPf99}: variables correspond to the
intuitionistic context, resources to the linear
context.\footnote{Note, however, that we impose no linearity
constraints in our type system!}

The inclusion, $\ret{\t}$, of terms into expressions has no analogue
in pure modal logic, but is specific to our interpretation of
memoization as a computational effect.  The typing rule for $\ret{\t}$
requires that $\t$ be resource-free to ensure that any dependence on
the argument to a memoized function is made explicit in the code
before computing the return value of the function.  In the first
instance, resources arise as parameters to memoized functions, with
further resources introduced by their incremental decomposition using
\textcd{let$\times$} and \textcd{mcase}.  These additional resources
track the usage of as-yet-unexplored parts of a data structure.
Ultimately, the complete value of a resource may be accessed using the
\textcd{let!} construct, which binds its value to a variable that may
be used without restriction.  In practice this means that those parts
of an argument to a memoized function on whose value the function
depends will be given modal type.  However, it is not essential that
all resources have modal type, nor that the computation depend upon
every resource that does have modal type.

\begin{figure}[h]
\small
\[
\begin{array}{|@{~~~~~~}c@{~~~~~~}|}
\hline
~ \\

\infer[\rlabel{variable}]  {\G ; \D  \ts x \cc \tau} {(\G(x)=\tau)}
\quad\quad
\infer[\rlabel{resource}]  {\G ; \D  \ts a \cc \tau} {(\D(a)=\tau)}
\\

\infer[\rlabel{number}]  {\G ; \D  \tsi n : \Int}{\strut} 
\quad\quad
\infer[\rlabel{unit}]  {\G ; \D  \tsi \uno : \tunit}{\strut} 
\\[3mm]

\infer [\rlabel{primitive}]
{\G ; \D  \ts \op{\ti,\ldots,\t_n} : \tau}
{\G ; \D  \ts \t_i : \tau_i ~~(1\leq i\leq n) ~&~ \tso \oper
:(\taui, \ldots, \tau_n)~\tau}
\\[3mm]

\infer[\rlabel{pair}]
{\G ; \D \ts \pair{\ti}{\tii} :\taui \cross \tauii}
{\G ; \D  \ts \ti : \taui & 
\G ; \D  \ts \tii : \tauii}
\\[3mm]

\infer[\rlabel{fun}]
{\G ; \D  \tsi \fun{f}{a}{\taui}{\tauii}{\e}  : \taui \ra \tauii}
{\G ,\, f \cc \taui \ra \tauii;\D,a\cc \taui \ts \e :\tauii}
\\[3mm]

\infer[\rlabel{fun value}]
{\G; \D \tsi \funval{l}{f}{a}{\taui}{\tauii}{\e} : \taui\ra\tauii}
{\G,f\cc\taui\ra\tauii; \D,a\cc\taui \ts \e:\tauii}
\\[3mm]

\infer [\rlabel{apply}]
{\G ; \D  \ts \apply{\ti}{\tii} : \tauii}
{\G ; \D \ts \ti: \taui \ra \tauii & \G ;  \D  \ts \tii : \taui}
\\[3mm]

\infer[\rlabel{bang}]
{\G ; \D  \ts \bang{\t} : \bang{\eta}}
{\G ; \emptyContext  \ts \t : \eta}
\\[3mm]

\infer[\rlabel{sum/inl}]
{\G ; \D  \ts \inl{\t}{\taui}{\tauii} : \taui \osum \tauii}
{\G ; \D  \ts \t : \taui} \qquad
\infer[\rlabel{sum/inr}]
{\G ; \D  \ts \inr{\t}{\taui}{\tauii} : \taui \osum \tauii}
{\G ; \D  \ts \t : \tauii}
\\[3mm]

\infer[\rlabel{roll}]
{\G ; \D  \ts \roll{\t} : \rt{\u}{\tau} }
{\G ; \D  \ts \t : [\rt{\u}{\tau}/\u]\tau}
\qquad
\infer[\rlabel{unroll}]
{\D  \ts \unroll{\t} :  [\rt{\u}{\tau}/\u]\tau}
{\G ; \D  \ts \t :\rt{\u}{\tau}}\\[3mm]

\hline
~\\
\infer[\rlabel{return}]
{\G ; \D  \ts \ret{\t} : \tau}
{\G ; \emptyContext  \ts \t : \tau}
\\[3mm]


\infer[\rlabel{let!}]
{ \G ; \D  \ts \letbin{x \cdcolon \eta}{\t}{\e} :\tau}
{ \G ; \D  \ts \t : \bang{\eta} &  \G, x \cc \eta; \D \ts \e : \tau}
\\[3mm]

\infer[\rlabel{let$\cross$}]
        {\G ; \D \ts \letpin{a_1\cc\taui}{a_2\cc\tauii}{\t}{\e} :\tau}
        {\G ; \D  \ts \t :\taui \cross \tauii & 
         \G ; \D,a_1 \cc \taui,a_2 \cc \tauii \ts \e : \tau}
\\[3mm]

\infer[\rlabel{case}]
{\G ; \D \ts  \caseofflat{\t}{\ai}{\taui}{\ei}{\aii}{\tauii}{\eii} : \tau}
{
\begin{array}{rcl}
\G ; \D & \ts & t : \taui \osum \tauii \\
\G ; \D,\ai\cc\taui & \ts & \ei : \tau\\
\G; \D,\aii\cc\tauii & \ts & \eii : \tau
\end{array}
}\\[4mm]

\hline
\end{array}
\]

\vspace{-2mm}
\caption{Typing judgments for terms (top) and expressions (bottom).}
\label{fig:static-semantics}
\end{figure}

The static semantics of \MFL{} consists of a set of rules for deriving
typing judgments of the form $\G;\D\ts t: \tau$, for terms, and
$\G;\D\ts e : \tau$, for expressions.  In these judgments $\G$ is a
\emph{variable type assignment}, a finite function assigning types to
variables, and $\D$ is a \emph{resource type assignment}, a finite
function assigning types to resources.  \figref{static-semantics}
shows the typing judgments for terms and expressions.

\subsection{Dynamic semantics}
\label{sec:mfl-dynamic}
\label{sec:dynamic-semantics}

\begin{figure}
\small
\[
\begin{array}{|@{~~~~~~}c@{~~~~~~}|}
\hline
~\\

\infer[\rlabel{unit}]
{\tis{\ms}{\uno} \treduces \tos{\uno}{\ms}}
{\strut}
\qquad
\infer[\rlabel{number}]
{\tis{\ms}{n} \treduces \tos{n}{\ms}}
{\strut}

\\[3mm]

\infer[\rlabel{primitive}]
{\tis{\ms}{\op{\ti,\ldots,t_n}} \treduces \tos{\OpApply{\oper}{(\vi,\ldots,\v_n}}{\ms_n}}
{
\tis{\ms}{\ti} & \treduces & \tos{\vi}{\ms_1} & \ldots & 
\tis{\ms_{n-1}}{\t_n} & \treduces & \tos{\v_n}{\ms_n}
}
\\[3mm]

\infer[\rlabel{pair}]
{\tis{\ms}{\pair{\ti}{\tii}} \treduces \tos{\pair{\vi}{\vii}}{\mspp}}
{
\begin{array}{rcl}
\tis{\ms}{\ti} & \treduces & \tos{\vi}{\msp}\\
\tis{\msp}{\tii} & \treduces & \tos{\vii}{\mspp}
\end{array}
}       
\\[3mm]


\infer[\rlabel{fun}]
{\tis{\ms}{\fun{f}{a}{\taui}{\tauii}{\e}} \treduces 
 \tos{\funval{l}{f}{a}{\taui}{\tauii}{\e}}{\ms[l \mapsto \emptyset]}
}
{
\begin{array}{c}
(l \not\in \dom{\ms})
\end{array}
}
\\[3mm]

\infer[\rlabel{fun val}]
{\tis{\ms}{\funval{l}{f}{a}{\taui}{\tauii}{\e}} \treduces\tos{\funval{l}{f}{a}{\taui}{\tauii}{\e}}{\ms}}
{
(l\in\dom{\ms})
}
\\[3mm]

\infer[\rlabel{apply}]
{\tis{\ms}{\apply{\ti}{\tii}} \treduces \tos{\v}{\msp}}
{
\begin{array}{rcl}
\tis{\ms}{\ti} & \treduces & \tos{\vi}{\msi} \\
\tis{\msi}{\tii} & \treduces & \tos{\vii}{\msii} \\
\eis{\msii}{l}{\nullbr}{[\vi,\vii/f,a]~e} & \ereduces & \eos{\v}{\msp} \\
\multicolumn{3}{c}{
(\vi = \funval{l}{f}{a}{\taui}{\tauii}{\e})
}
\end{array}
}       
\\[3mm]

\infer[\rlabel{bang}]
{\tis{\ms}{\bang{\t}} \treduces \tos{\bang{\v}}{\msp}}
{\tis{\ms}{\t} \treduces \tos{\v}{\msp}}
\\[3mm]

\infer[\rlabel{case/inl}]
{\tis{\ms}{\inl{\t}{\taui}{\tauii}} \treduces \tos{\inl{\v}{\taui}{\tauii}}{\msp}}
{\tis{\ms}{\t} \treduces \tos{\v}{\msp}}
\qquad
\infer[\rlabel{case/inr}]
{\tis{\ms}{\inr{\t}{\taui}{\tauii}} \treduces \tos{\inr{\v}{\taui}{\tauii}}{\msp}}
{\tis{\ms}{\t} \treduces \tos{\v}{\msp}}
\\[3mm]

\infer[\rlabel{roll}]
{\tis{\ms}{\roll{\t}} \treduces \tos{\roll{\v}}{\msp}}
{\tis{\ms}{\t} \treduces \tos{\v}{\msp}}
\qquad
\infer[\rlabel{unroll}]
{\tis{\ms}{\unroll{\t}} \treduces \tos{\v}{\msp}}
{\tis{\ms}{\t} \treduces \tos{\roll{\v}}{\msp}}\\[4mm]





\hline
\end{array}
\]

\vspace*{-4mm}
\caption{Evaluation of terms.}
\label{fig:term-dynamic}

\end{figure}

\begin{figure}
\small
\[
\begin{array}{|@{~~~~}c@{~~~~}|}
\hline
~\\

\infer  [\rlabel{return/found}]
{\eis{\ms}{l}{\br}{\ret{\t}} \ereduces \eos{\v}{\ms}}
{\ms(l)(\br)=\v} 
\\[3mm]

\infer  [~\rlabel{return/not found}]
{\eis{\ms}{l}{\br}{\ret{\t}} \ereduces \eos{\v}{\msp[l \leftarrow \theta'[\br\mapsto\v]]}}
{       
\begin{array}{c}
\ms(l) = \theta \quad \theta(\br)\uparrow \\
\tis{\ms}{\t}  \treduces \tos{\v}{\msp}\\
\msp(l) = \theta'
\end{array}
}
\\[3mm]

\infer[\rlabel{let!}]
{\eis{\ms}{l}{\br}{\letbin{x:\eta}{\t}{\e}} \ereduces  \eos{\vp}{\mspp}}
{\begin{array}{rcl} 
\tis{\ms}{\t} & \treduces & \tos{\bang{\v}}{\msp}\\
\eis{\msp}{l}{\consbr{\bangev{\v}}{\br}}{[\v/x]e} & \treduces & \tos{\vp}{\mspp}
\end{array}}
\\[3mm]

\infer[\rlabel{let$\cross$}]
{\eis{\ms}{l}{\br}{\letpin{a_1}{a_2}{\t}{\e}} \treduces  \eos{\v}{\mspp}}
{\begin{array}{rcl}     
\tis{\ms}{\t} & \treduces  & \tos{\vi \cross \vii}{\msp}\\
\eis{\msp}{l}{\br}{[\vi/a_1,\vii/a_2]e} & \ereduces  & \eos{\v}{\mspp}\\
\end{array} 
} 
\\[3mm]

\infer[\rlabel{case/inl}]
{\eis{\ms}{l}{\br}{\caseofflat{\t}{\ai}{\taui}{\ei}{\aii}{\tauii}{\eii}} \treduces \eos{\vi}{\mspp}}
{
\begin{array}{rcl}
\tis{\ms}{\t}  & \treduces & \tos{\inl{\v}{\taui}{\tauii}}{\msp} \\
\eis{\msp}{l}{\consbr{\inlev}{\br}}{[\v/\ai]\ei} & \ereduces & \eos{\vi}{\mspp} 
\end{array}
}
\\[3mm]

\infer[\rlabel{case/inr}]
{\eis{\ms}{l}{\br} {\caseofflat{\t}{\ai}{\taui}{\ei}{\aii}{\tauii}{\eii}}\treduces \eos{\vii}{\mspp}
}
{
\begin{array}{rcl}
\tis{\ms}{\t}  & \treduces & \tos{\inr{\v}{\taui}{\tauii}}{\msp} \\
\eis{\msp}{l}{\consbr{\inrev}{\br}}{[\v/\aii]\eii} & \ereduces & \eos{\vii}{\mspp} 
\end{array}
}\\[4mm]

\hline

\end{array}
\]

\vspace*{-4mm}
\caption{Evaluation of expressions.}
\label{fig:expr-dynamic}
\end{figure}

The dynamic semantics of \MFL{} formalizes selective memoization.
Evaluation is parameterized by a store containing memo tables that
track the behavior of functions in the program.  Evaluation of a
function expression allocates an empty memo table and associates it
with the function.  Application of a memoized function is affected by,
and may affect, its memo table.  When the function value becomes
inaccessible, so is its associated memo table and the storage required
for both can be reclaimed.

Unlike conventional memoization, however, the memo table is keyed by
control flow information rather than by the values of arguments to
memoized functions.  This is the key to supporting selective
memoization.  Expression evaluation is essentially an exploration of
the available resources culminating in a resource-free term that
determines its value.  Since the exploration is data-sensitive, only
certain aspects of the resources may be relevant to a particular
outcome.  For example, a memoized function may take a pair of integers
as argument, with the outcome determined independently of the second
component in the case that the first is positive.  By recording
control-flow information during evaluation, we may use it to provide
selective memoization.

For example, in the situation just described, all pairs of the form
$\pair{0}{\v}$ should map to the same result value, irrespective of the value
$\v$.  In conventional memoization the memo table would be keyed by the pair,
with the result that redundant computation is performed in the case that the
function has not previously been called with $\v$, even though the value of
$\v$ is irrelevant to the result!  In our framework we instead key the memo
table by a ``branch'' that records sufficient control flow information to
capture the general case.  Whenever we encounter a \textcd{return} statement,
we query the memo table with the current branch to determine whether this
result has been computed before.  If so, we return the stored value; if not, we
evaluate the \textcd{return} statement, and associate that value with that
branch in the memo table for future use.  It is crucial that the returned term
not contain any resources so that we are assured that its value does not change
across calls to the function.

The dynamic semantics of \MFL{} is given by a set of rules for deriving
judgments of the form $\tis{\ms}{\t}\treduces\tos{\v}{\msp}$ (for terms) and
$\eis{\ms}{l}{\br}{\e}\ereduces\eos{\v}{\msp}$ (for expressions).  The rules
for deriving these judgments are given in Figures~\ref{fig:term-dynamic}
and~\ref{fig:expr-dynamic}.  These rules make use of branches, memo tables, and
stores, whose precise definitions are as follows.

A \emph{simple branch} is a list of \emph{simple events} corresponding to
``choice points'' in the evaluation of an expression.
\begin{displaymath}
  \begin{array}{lrcl}
    \textit{Simple Event}  & \ev & \bnfdef & \bangev{\v} \bnfalt \inlev \bnfalt
    \inrev \\
    \textit{Simple Branch} & \br & \bnfdef & \nullbr \bnfalt \consbr{\ev}{\br}
  \end{array}
\end{displaymath}
We write $\extbr{\br}{\ev}$ to stand for the extension of $\br$ with the event
$\ev$ at the end.

A \emph{memo table}, $\theta$, is a finite function mapping simple
branches to values.  We write $\theta[\br\mapsto\v]$, where
$\br\notin\dom{\theta}$, to stand for the extension of $\theta$ with
the given binding for $\br$.  We write $\theta(\br)\uparrow$ to mean
that $\beta\notin\dom{\theta}$.

A \emph{store}, $\ms$, is a finite function mapping \emph{locations}, $l$, to
memo tables.  We write $\ms[l\mapsto\theta]$, where $l\notin\dom{\ms}$, to
stand for the extension of $\ms$ with the given binding for $l$.  When
$l\in\dom{\ms}$, we write $\ms[l\leftarrow\theta]$ for the store $\ms$ that
maps $l$ to $\theta$ and $l'\not=l$ to $\ms(l')$.

Term evaluation is largely standard, except for the evaluation of (memoizing)
functions and applications of these to arguments.  Evaluation of a memoizing
function term allocates a fresh memo table, which is then associated with the
function's value.  Expression evaluation is initiated by an application of a
memoizing function to an argument.  The function value determines the memo
table to be used for that call.  Evaluation of the body is performed relative
to that table, initiating with the null branch.

Expression evaluation is performed relative to a ``current'' memo table and
branch.  When a \textcd{return} statement is encountered, the current memo
table is consulted to determine whether or not that branch has previously been
taken.  If so, the stored value is returned; otherwise, the argument term is
evaluated, stored in the current memo table at that branch, and the value is
returned.  The \textcd{let!} and \textcd{mcase} expressions extend the current
branch to reflect control flow.  Since \textcd{let!} signals dependence on a
complete value, that value is added to the branch.  Case analysis, however,
merely extends the branch with an indication of which case was taken.  The
\textcd{let$\times$} construct does not extend the branch, because no
additional information is gleaned by splitting a pair.

\subsection{Soundness of \MFL{}}
\label{sec:mfl::soundness}
We prove the soundness of \MFL{} relative to a non-memoizing semantics
for the language.  It is straightforward to give a purely functional
semantics to \MFL{} by an inductive definition of the relations
$\t\treducespure \v$ and $\e\ereducespure\v$, where $\v$ is a
\emph{pure value} with no location subscripts (see, for
example,~\citep{pfenning+:judgmental}).  We show that memoization does
not affect the outcome of evaluation as compared to the non-memoized
semantics (\thmref{mfl-is-sound}).  To make this precise, we must
introduce some additional machinery.

The \emph{underlying term}, $\ers{\t}$, of a term, $\t$, is obtained by erasing
all location subscripts on function values occurring within $\t$.  The
\emph{underlying expression}, $\ers{\e}$, of an expression, $\e$, is defined in
the same way.  As a special case, the \emph{underlying value}, $\ers{\v}$, of a
value, $\v$, is the underlying term of $\v$ regarded as a term.  It is easy to
check that every pure value arises as the underlying value of some impure
value.  Note that passage to the underlying term or expression obviously
commutes with substitution.  The \emph{underlying branch}, $\ers{\beta}$, of a
simple branch, $\beta$, is obtained by replacing each event of the form
$\bang{\v}$ in $\beta$ by the corresponding underlying event,
$\bang{(\ers{\v})}$.

\smallskip

The partial \emph{access functions}, $\atbr{\t}{\br}$ and $\atbr{\e}{\br}$,
where $\br$ is a simple branch, and $\t$ and $e$ are variable-free (but not
necessarily resource-free), are defined as follows.  The definition may be
justified by lexicographic induction on the structure of the branch followed by
the size of the expression.
\begin{displaymath}
  \begin{array}{rcl}
    \atbr{\t}{\br} & = & \atbr{\e}{\br}
    \\
    \multicolumn{3}{r}{
      (\textit{where}\ \t=\fun{f}{a}{\taui}{\tauii}{\e})
    }
    \\ \\
    \atbr{\ret{\t}}{\nullbr} & = & \ret{\t} \\
    \atbr{\letbin{x\cdcolon\tau}{\t}{\e}}{\bangbr{\br}{\v}} & = & \atbr{\subst{v}{x}{\e}}{\br} \\
    \atbr{\letpin{a_1\cdcolon\taui}{a_2\cdcolon\tauii}{\t}{\e}}{\br} & = &
    \atbr{\e}{\br} \\
    \atbr{\caseofflat{\t}{a_1}{\taui}{\e_1}{a_2}{\tauii}{\e_2}}{\inlbr{\br}} &
    = & \atbr{\e_1}{\br} \\
    \atbr{\caseofflat{\t}{a_1}{\taui}{\e_1}{a_2}{\tauii}{\e_2}}{\inrbr{\br}} &
    = & \atbr{\e_2}{\br} \\
  \end{array}
\end{displaymath}
This function will only be of interest in the case that $\atbr{\e}{\br}$ is a
\textcd{return} expression, which, if well-typed, cannot contain free
resources.  Note that $\ers{(\atbr{\e}{\br})} = \atbr{\ers{\e}}{\ers{\br}}$,
and similarly for values, $\v$.

We are now in a position to justify a subtlety in the second \textcd{return}
rule of the dynamic semantics, which governs the case that the returned value
has not already been stored in the memo table.  This rule extends, rather than
updates, the memo table with a binding for the branch that determines this
\textcd{return} statement within the current memoized function.  But why, after
evaluation of $\t$, is this branch undefined in the revised store, $\msp$?  If
the term $\t$ were to introduce a binding for $\br$ in the memo table
$\sigma(l)$, it could only do so by evaluating the very same \textcd{return}
statement, which implies that there is an infinite loop, contradicting the
assumption that the \textcd{return} statement has a value, $\v$.
\begin{lemma}
\label{lemma:at-undef}
  If $\tis{\ms}{\t}\treduces\tos{v}{\msp}$, $\atbr{\sigma(l)}{\br}=\ret{\t}$,
  and $\ms(l)(\br)$ is undefined, then $\msp(l)(\br)$ is also undefined.
\end{lemma}

\smallskip

An \emph{augmented branch}, $\gbr$, is an extension of the notion of branch in
which we record the bindings of resource variables.  Specifically, the argument
used to call a memoized function is recorded, as are the bindings of resources
created by pair splitting and case analysis.  Augmented branches are
inductively defined by the following grammar:
\begin{displaymath}
  \begin{array}{lrcl}
    \textit{Augmented Event} & \gev & \bnfdef & \callgev{\v} \bnfalt
    \banggev{\v} \bnfalt \pairgev{\vi}{\vii} \bnfalt \inlgev{\v} \bnfalt
    \inrgev{\v} \\ 
    \textit{Augmented Branch} & \gbr & \bnfdef & \nullgbr \bnfalt
    \consgbr{\gev}{\gbr} \\
  \end{array}
\end{displaymath}
We write $\extgbr{\gbr}{\gev}$ for the extension of $\gbr$ with $\gev$ at the
end.  There is an obvious \emph{simplification} function, $\simp{\gbr}$, that
yields the simple branch corresponding to an augmented branch by dropping
``call'' events, $\callgev{\v}$, and ``pair'' events, $\pairgev{\vi}{\vii}$,
and by omitting the arguments to ``injection'' events, $\inlgev{\v}$,
$\inrgev{\v}$.  The \emph{underlying augmented branch}, $\ers{\gbr}$,
corresponding to an augmented branch, $\gbr$, is defined by replacing each
augmented event, $\gev$, by its corresponding underlying augmented event,
$\ers{\gev}$, which is defined in the obvious manner.  Note that
$\ers{(\simp{\gbr})} = \simp{(\ers{\gbr})}$.

The partial access functions $\atgbr{\e}{\gbr}$ and $\atgbr{\t}{\gbr}$ are
defined for closed expressions $e$ and closed terms $t$ by the following
equations:
\begin{displaymath}
  \begin{array}{rcl}
    \atgbr{\t}{\callgbr{\gbr}{\v}} & = & \atgbr{\subst{\t,\v}{f,a}{\e}}{\gbr}
    \\
    \multicolumn{3}{r}{
      (\textit{where}\ \t=\fun{f}{a}{\taui}{\tauii}{\e})
    }
    \\ \\
    \atgbr{\e}{\nullgbr} & = & e \\
    \atgbr{\letbin{x\cdcolon\tau}{\t}{\e}}{\banggbr{\gbr}{\v}} & = &
    \atgbr{\subst{\v}{x}{\e}}{\gbr} \\
    \atgbr{\letpin{a_1\cdcolon\taui}{a_2\cdcolon\tauii}{\t}{\e}}{\splitgbr{\br}{\vi}{\vii}} & = & \atgbr{\subst{\vi,\vii}{a_1,a_2}{\e}}{\br} \\
    \atgbr{\caseofflat{\t}{a_1}{\taui}{\e_1}{a_2}{\tauii}{\e_2}}{\inlgbr{\br}{\v}} & = & \atgbr{\subst{\v}{a_1}{\e_1}}{\br} \\
    \atgbr{\caseofflat{\t}{a_1}{\taui}{\e_1}{a_2}{\tauii}{\e_2}}{\inrgbr{\br}{\v}} & = & \atgbr{\subst{\v}{a_2}{\e_2}}{\br} \\
  \end{array}
\end{displaymath}
Note that $\ers{(\atgbr{\e}{\gbr})} = \atgbr{\ers{\e}}{\ers{\gbr}}$, and
similarly for values, $\v$.

Augmented branches, and the associated access function, are needed for the
proof of soundness.  The proof maintains an augmented branch that enriches the
current simple branch of the dynamic semantics.  The additional information
provided by augmented branches is required for the induction, but it does not
affect any \textcd{return} statement it may determine.
\begin{lemma}
\label{lemma:at-simp}
  If $\atgbr{\e}{\gbr}=\ret{\t}$, then $\atbr{\e}{\simp{\gbr}}=\ret{\t}$.
\end{lemma}

A \emph{function assignment}, $\MS$, is a finite mapping from locations to
well-formed, closed, pure function values.  A function assignment is
\emph{consistent with} a term, $t$, or expression, $e$, if and only if whenever
$\funval{l}{f}{a}{\taui}{\tauii}{e}$ occurs in either $t$ or $e$, then
$\MS(l)=\fun{f}{a}{\taui}{\tauii}{\ers{e}}$.  Note that if a term or expression
is consistent with a function assignment, then no two function values with
distinct underlying values may have the same label.  A function assignment is
consistent with a store, $\ms$, if and only if whenever $\ms(l)(\br)=\v$, then $\MS$ is
consistent with $\v$.

A store, $\ms$, \emph{tracks} a function assignment, $\MS$, if and only if $\MS$ is
consistent with $\sigma$, $\dom{\ms}=\dom{\MS}$, and for every $l\in\dom{\ms}$,
if $\ms(l)(\br)=\v$, then
\begin{enumerate}
\item $\atbr{\MS(l)}{\ers{\br}}=\ret{\ers{\t}}$,
\item $\ers{\t}\treducespure\ers{\v}$,
\end{enumerate}
Thus if a branch is assigned a value by the memo table associated with a
function, it can only do so if that branch determines a \textcd{return}
statement whose value is the assigned value of that branch, relative to the
non-memoizing semantics.

We are now in a position to prove the soundness of \MFL{}.
\begin{theorem}
\label{thm:soundness-induction}

  \begin{enumerate}
  \item If $\tis{\ms}{\t}\treduces\tos{\v}{\msp}$, $\MS$ is consistent with
    $\t$, $\ms$ tracks $\MS$, $\emptyContext;\emptyContext\ts \t:\tau$, then
    $\ers{\t}\treducespure \ers{\v}$ and there exists $\MSP\supseteq\MS$ such
    that $\MSP$ is consistent with $\v$ and $\msp$ tracks $\MSP$.
  \item If $\eis{\ms}{l}{\br}{\e}\ereduces\eos{\v}{\msp}$, $\MS$ is consistent
    with $e$, $\ms$ tracks $\MS$, $\simp{\gbr}=\br$,
    $\atgbr{\MS(l)}{\ers{\gbr}}=\ers{\e}$, and $\emptyContext;\emptyContext\ts
    \e:\tau$, then there exists $\MSP\supseteq\MS$ such that
    $\ers{\e}\ereducespure\ers{\v}$, $\MSP$ is consistent with $\v$, and $\msp$
    tracks $\MSP$.
  \end{enumerate}
\end{theorem}
\begin{proof}
  The proof proceeds by simultaneous induction on the memoized evaluation
  relation.  We consider here the five most important cases of the proof:
  function values, function terms, function application terms, and return
  expressions.
  
  For function values $\t=\funval{l}{f}{a}{\taui}{\tauii}{\e}$, simply take
  $\MSP=\MS$ and note that $\v=\t$ and $\msp=\ms$.
  
  For function terms $\t=\fun{f}{a}{\taui}{\tauii}{\e}$, note that
  $\v=\funval{l}{f}{a}{\taui}{\tauii}{\e}$ and $\msp=\ms[l\mapsto\emptyset]$,
  where $l\notin\dom{\ms}$.  Let $\MSP=\MS[l\mapsto \ers{\v}]$, and note that
  since $\ms$ tracks $\MS$, and $\ms(l)=\emptyset$, it follows that $\msp$
  tracks $\MSP$.  Since $\MS$ is consistent with $\t$, it follows by
  construction that $\MSP$ is consistent with $\v$.  Finally, since
  $\ers{\v}=\ers{\t}$, we have $\ers{\t}\treducespure\ers{\v}$, as required.
  
  For application terms $\t=\apply{\ti}{\tii}$, we have by induction that
  $\ers{\ti}\treducespure\ers{\vi}$ and there exists $\MSi\supseteq\MS$
  consistent with $\vi$ such that $\msi$ tracks $\MSi$.  Since
  $\vi=\funval{l}{f}{a}{\taui}{\tauii}{\e}$, it follows from consistency that
  $\MSi(l)=\ers{\vi}$.  Applying induction again, we obtain that
  $\ers{\tii}\treducespure\ers{\vii}$, and there exists $\MSii\supseteq\MSi$
  consistent with $\vii$ such that $\msii$ tracks $\MSii$.  It follows that
  $\MSii$ is consistent with $\subst{\vi,\vii}{f,a}{\e}$.  Let
  $\gbr=\consgbr{\callgev{\vii}}{\nullgbr}$.  Note that
  $\simp{\gbr}=\nullbr=\br$ and we have
  \begin{displaymath}
    \begin{array}{rcl}
      \atgbr{\MSii(l)}{\ers{\gbr}} & = & \atgbr{\ers{\vi}}{\ers{\gbr}} \\
      & = & \ers{(\atgbr{\vi}{\gbr})} \\
      & = & \ers{(\subst{\vi,\vii}{f,a}{\e})} \\
      & = & \subst{\ers{\vi},\ers{\vii}}{f,a}{\ers{\e}}.
    \end{array}
  \end{displaymath}
  Therefore, by induction,
  $\subst{\ers{\vi},\ers{\vii}}{f,a}{\ers{\e}}\ereducespure\ers{\vp}$, and
  there exists $\MSP\supseteq\MSii$ consistent with $\vp$ such that $\msp$
  tracks $\MSP$.  It follows that $\ers{(\apply{\ti}{\tii})} =
  \apply{\ers{\ti}}{\ers{\tii}}\treducespure\ers{\vp}$, as required.
  
  For return statements, we have two cases to consider, according to whether
  the current branch is in the domain of the current memo table.  Suppose that
  $\eis{\ms}{l}{\br}{\ret{\t}}\ereduces\eos{\v}{\msp}$ with $\MS$ consistent
  with $\ret{\t}$, $\ms$ tracking $\MS$, $\simp{\gbr}=\br$,
  $\atgbr{\MS(l)}{\ers{\gbr}}=\ers{(\ret{\t})}=\ret{\ers{\t}}$, and
  $\emptyContext;\emptyContext\ts \ret{\t}:\tau$.  Note that by
  Lemma~\ref{lemma:at-simp}, $\ers{(\atbr{\MS(l)}{\br})} =
  \atbr{\MS(l)}{\ers{\br}} = \ret{\ers{\t}}$.
  
  For the first case, suppose that $\ms(l)(\br)=\v$.  Since $\ms$ tracks $\MS$
  and $l\in\dom{\ms}$, we have $\MS(l)=\fun{f}{a}{\taui}{\tauii}{\ers{\e}}$
  with $\atbr{\ers{\e}}{\ers{\br}}=\ret{\ers{\t}}$, and
  $\ers{\t}\treducespure\ers{\v}$.  Note that $\msp=\ms$, so taking $\MSP=\MS$
  completes the proof.
  
  For the second case, suppose that $\ms(l)(\br)$ is undefined.  By induction
  $\ers{\t}\treducespure\ers{\v}$ and there exists $\MSP\supseteq\MS$
  consistent with $\v$ such that $\msp$ tracks $\MSP$.  Let $\theta'=\msp(l)$,
  and note $\theta'(\br)\uparrow$, by Lemma~\ref{lemma:at-undef}.  Let
  $\theta''=\theta'[\br\mapsto\v]$, and $\mspp=\msp[l\leftarrow\theta'']$.  Let
  $\MSPP=\MSP$; we are to show that $\MSPP$ is consistent with $\v$, and
  $\mspp$ tracks $\MSPP$.  By the choice of $\MSPP$ it is enough to show that
  $\atbr{\MSP(l)}{\ers{\br}}=\ret{\ers{\t}}$, which we noted above.
  
\end{proof}

The soundness theorem (\thmref{mfl-is-sound}) for \MFL{} states that
evaluation of a program (a closed term) with memoization yields the
same outcome as evaluation without memoization.  The theorem follows
from \thmref{soundness-induction}.
\begin{theorem}[Soundness]
\label{thm:mfl-is-sound}
  If $\tis{\emptyset}{\t}\treduces\tos{\v}{\ms}$, where
  $\emptyset;\emptyset\ts \t:\tau$, then $\ers{\t}\treducespure \ers{\v}$.
\end{theorem}

Type safety follows from the soundness theorem, since type safety holds for the
non-memoized semantics.  In particular, if a term or expression had a
non-canonical value in the memoized semantics, then the same term or expression
would have a non-canonical value in the non-memoized semantics, contradicting
safety for the non-memoized semantics.

\subsection{Asymptotic complexity}
\label{sec:mfl::complexity}

We show that memoization slows down an \mfl program by a constant
factor (expected) with respect to a standard, non-memoizing semantics
even when no results are re-used.  The result relies on representing a
branch as a sequence of integers and using this sequence to key memo
tables, which are represented with nested hash tables.

To represent branches as integer sequences we use the property of \mfl
that the underlying type $\eta$ of a bang type, $\bang{\eta}$, is an
indexable type. Since any value of an indexable type has an integer
index, we represent a branch as sequence of integers corresponding to
the indices of \nletbang'ed values, and zero or one for \textcd{inl}
and \textcd{inr}.

We represent memo tables as nested hash tables.  A nested hash table
is a tree of hash tables consisting of internal hash tables and
external hash tables (leaves).  Internal hash tables map an integer
(an index) to another hash table.  External hash tables map an integer
to the result of the function.  Given a branch $\beta$ of length $m$
(consisting of $m$ indices), a lookup proceeds by indexing each key in
order starting at the root of the nested hash table.  Each lookup
except for the last returns a hash table, which is then used for the
next lookup with the next index.  The last lookup returns the desired
result in the case of a memo hit, or nothing in the case of a memo
miss.  Since a lookup takes expected constant time, a lookup with a
branch of length $m$ takes $O(m)$ time.  The same bounds holds for
update operations (insertions, deletions).

\begin{theorem}
  The overhead of an \mfl program with respect to a pure,
  non-memoizing semantics is expected $O(1)$, where the expectation is
  over internal randomization used for hash tables. 
\end{theorem}

\begin{proof}
  Consider a non-memoizing semantics, where the \textcd{return} rule
  always evaluates its body and neither looks up nor updates memo
  tables (stores). Consider an \MFL{} program and let $T$ denote the
  time (the number of evaluation steps) it takes to evaluate the
  program with respect to this non-memoizing semantics.  Let $T'$
  denote the time it takes to evaluate the same program with respect
  to the memoizing semantics.  In the worst case, no results are
  re-used, thus the difference between $T$ and $T'$ is due to
  memo-table operations (lookups and updates) performed by the
  memoizing semantics.

  To bound the time for memo table operations, consider a memo-table
  operation with a branch $\br$ and let $m$ be the length of the
  branch.  With nested hash tables, the operation requires expected
  $\Theta(m)$ time.  Since the non-memoizing semantics takes
  $\Theta(m)$ time to build the branch, the overhead of the memo-table
  operations is expected $O(1)$.  Since a branch is used to perform
  only a constant number of memo-table operations (one lookup and one
  update) we conclude that overhead of selective memoization is $O(1)$
  in expectation.
\end{proof}

\section{Implementation}
\label{sec:implementation}
\label{sec:imp}
We describe an implementation of the \mfl language as a Standard ML
library.  Since the library cannot differentiate between resources and
variables syntactically, it uses a separate type for resources.  The
library therefore cannot enforce statically the aspects of \mfl that
rely on the syntactic distinction between resources; instead it
employs run-time checks to detect violations of correct
usage.\footnote{We describe elsewhere a library for Standard ML that
  can in fact enforce the \mfl type system
  statically~\citep{AcarBlBlHaTa06}.  The approach, however, does not
  scale well.}

\begin{figure}
\centering{\fbox{
\parbox[t]{1.5in}{
\begin{codeListingS}
XX\=XX\=XX\=XX\=XX\=XX\=XX\=XX\=XX\=XX\=XX\=XX\=XXXXXXXXXXXXXXXXXXX\=\kill
signature MEMO = \\ 
sig \\
\>(* Expressions *)\\
\>type 'a expr\\
\>val return:(unit -> 'a) -> 'a expr\\[2mm]

\>(* Resources *)\\
\>type 'a res~ \\
\>val expose:'a res -> 'a\\[2mm]

\> (* Bangs *)\\
\>type 'a bang\\
\>val bang:('a -> int) -> 'a -> 'a bang\\
\>val letBang:('a bang) -> ('a -> 'b expr) -> 'b expr\\[2mm]

\>(* Products *)\\
\>type ('a,'b) prod\\
\>val pair:'a -> 'b -> ('a,'b) prod\\
\>val letx:('a,'b) prod -> (('a res * 'b res) -> 'c expr) -> 'c expr\\
\>val split:('a,'b) prod -> (('a * 'b) -> 'c) -> 'c \\[2mm]

\>(* Sums *)\\
\>type ('a,'b) sum\\
\>val inl:'a -> ('a,'b) sum\\
\>val inr:'b -> ('a,'b) sum\\
\>val mcase:('a,'b) sum -> ('a res -> 'c expr) -> ('b res -> 'c expr) -> 'c expr\\
\>val choose:('a,'b) sum -> ('a -> 'c) -> ('b -> 'c) -> 'c \\[2mm]

\>(* Memoized arrow *)\\
\>type ('a,'b) marrow\\
\>val mfun:('a res -> 'b expr) -> ('a,'b) marrow\\
\>val mfun\_rec:(('a, 'b) marrow -> 'a res -> 'b expr) -> ('a,'b) marrow\\
\>val mapply:('a,'b) marrow -> 'a -> 'b\\
end\\
\>\\
signature BOX = \\
sig\\
\>type 'a box\\[2mm]
\>val box:'a -> 'a box\\
\>val unbox:'a box -> 'a\\
\>val keyOf:'a box -> int\\
end
\end{codeListingS}}}}

\caption{The signatures for the memo library and boxes.}
\label{fig:sig}
\end{figure}

The interface for the library (\figref{sig}) provides types for
expressions, resources, bangs, products, sums, memoized functions
along with their introduction and elimination forms.  All expressions
have type {\tt 'a expr}, which is a monad with {\tt return} as the
inclusion and various forms of bind operations as elimination forms
{\tt letBang, letx}, and {\tt mcase}.  A resource has type~{\tt 'a
  res}.  The library provides no explicit introduction form for
resources.  Instead, resources are created by {\tt letx}, {\tt mcase},
{\tt mfun\_rec}, and {\tt mfun} primitives.  The elimination form for
resources is {\tt expose} which returns the underlying value of a
resources. 

The introduction and elimination form for bang types are {\tt bang}
and {\tt letBang}.  The introduction and elimination form for product
types are {\tt pair}, and {\tt letx} and {\tt split} respectively.
The {\tt letx} is a bind operation for the monad {\tt expr}; {\tt
  split} is the elimination form for the term context.  The treatment
of sums is similar to product types.  The introduction forms are {\tt
  inl} and {\tt inr}, and the elimination forms are {\tt mcase} and
{\tt choose}; {\tt mcase} is a bind operation for the {\tt expr} monad
and {\tt choose} is the elimination for the term context.


The {\tt mfun} and {\tt mfun\_rec} primitives introduce memoized
functions. The {\tt mfun} primitive takes a function of type \mbox{\tt
  'a res -> 'b expr} and returns the memoized function of type
\mbox{\tt ('a,'b) marrow}; {\tt mfun\_rec} is similar to {\tt mfun}
but it also takes its memoized version as argument.  Note that the
result type does not contain the ``effect'' {\tt expr}---the library
encapsulate memoization effects, which are benign, within the
function.  The elimination form for the {\tt marrow} is the memoized
apply function {\tt mapply}.

In addition to primitives for memoization, the library provides a
facilities for boxing and unboxing of values.  As described in
\secref{framework} boxes enables injecting ordinary types into
indexable types.  \figref{sig} shows the signature for boxes. 

\begin{figure}
\centering{\fbox{
\parbox[t]{1.5in}{
\begin{codeListingS}
XX\=\kill
signature MEMO\_TABLE = \\
sig\\
\>type 'a memotable\\

\>val empty: unit -> 'a memotable\\
\>val extend: 'a memotable -> int list -> ('a option * 'a memotable option)\\
\>val insert: 'a -> 'a memotable -> unit\\
end
\end{codeListingS}}}}
\caption{The signatures for memo tables.}
\label{fig:memo-table}
\end{figure} 

\begin{figure}[h]
\centering{\fbox{
\parbox[t]{1.5in}{
\begin{codeListingS}
XX\=XX\=XX\=XX\=XX\=XX\=XX\=XX\=XX\=XX\=XX\=XX\=XXXXXXXXXXXXXXXXXXX\=\kill
functor BuildMemo (structure MemoTable:MEMO\_TABLE):MEMO = \\
struct\\
\>type 'a expr = int list * (unit -> 'a)\\
\>fun return f = (nil,f)\\[2mm]

\>type 'a res = 'a\\
\>fun res v = v\\
\>fun expose r = r\\[2mm]

\>type 'a bang = 'a * ('a -> int)\\
\>fun bang h t = (t,h)\\
\>fun letBang b f = \\
\>\>let val (v,h) = b\\
\>\>~~~~val (branch,susp) = f v\\
\>\>in ((h v)::branch, susp) end\\[2mm]

\>type ('a,'b) prod = 'a * 'b\\
\>fun pair x y = (x,y)\\
\>fun split p f = f p\\
\>fun letx (p as (x1,x2)) f = f (res x1, res x2) \\[2mm]

\>datatype ('a,'b) sum = INL of 'a | INR of 'b\\
\>fun inl v = INL(v)\\
\>fun inr v = INR(v)\\
\>fun mcase s f g = \\
\>\>let val (lr,(branch,susp)) = case s of\\
\>\>\>\>~~~~~~~~~~~~~~~~~~~~~~~~~~~INL v => (0,f (res v))\\
\>\>\>~~~~~~~~~~~~~~~~~~~~~~~~~~~|~INR v => (1,g (res v))\\
\>\>in \\
\>\>\>(lr::branch,susp) \\
\>\>end\\
\>fun choose s f g = case s of~INL v => f v~|~INR v => g v\\[2mm]

\>type ('a,'b) marrow = 'a -> 'b   \\
\>fun mfun\_rec f = \\
\>\>let val mtable = MemoTable.empty ()\\
\>\>\>\>fun mf rf x = \\
\>\>\>\>\>let val (branch,susp) = f rf (res x)\\
\>\>\>\>\>\>~~val result = case MemoTable.extend mtable branch  of\\
~~~~~~~~~~~~~~~~~~~~~~~~~~~~~(NONE,SOME mtable') =>~~~~~~~~(* Not found *)\\
~~~~~~~~~~~~~~~~~~~~~~~~~~~~~~~let val v = susp ()\\
~~~~~~~~~~~~~~~~~~~~~~~~~~~~~~~~~~~val \_ = MemoTable.insert v mtable'\\
~~~~~~~~~~~~~~~~~~~~~~~~~~~~~~~in v end \\
~~~~~~~~~~~~~~~~~~~~~~~~~~~~|~(SOME v,NONE) => v~~~~~~~~~~(* Found *)\\
\>\>\>\>\>in result end\\
\>\>\>\>fun mf' x = mf mf' x\\
\>\>in\\
\>\>\>mf'\\ 
\>\>end\\[2mm]

\>fun mfun f = ... (* Similar to mfun\_rec *)\\[2mm]

\>fun mapply f v = f v\\[2mm]
end
\end{codeListingS}}}}
\caption{The implementation of the memoization library.}

\label{fig:implementation}
\end{figure}

The library implements memo tables as nested
hash tables as described in \secref{mfl::complexity}.
\figref{memo-table} shows the interface for the memo tables.  The
\ttt{empty} function returns an empty memo table.  The \ttt{extend}
function performs a look up with the provided int (index) list and
returns a pair consisting of the result (if found) and the extended
memo table.  The \ttt{insert} function inserts the provided result
into the specified memo table.

\figref{implementation} shows an implementation of the library.  The
{\tt bang} primitive takes a value and an injective function, called
the index function, that maps the value to an integer, called the
index.  The index of a value is used to key memo tables.  The
restriction that the indices be unique enables implementing memo
tables using hashing.  The primitive {\tt letBang} takes a value {\tt
  b} of bang type and a body.  It applies the body to the underlying
value of {\tt b}, and extends the branch with the index of {\tt b}.
The function {\tt letx} takes a pair {\tt p} and a body. It binds the
parts of the pair to two resources and and applies the body to the
resources; as with the operational semantics, {\tt letx} does not
extend the branch. The function {\tt mcase} takes value {\tt s} of sum
type and a body. It branches on the outer form of {\tt s} and binds
its inner value to a resource.  It then applies the body to the
resource and extends the branch with {\tt 0} or {\tt 1} depending on
the outer form of {\tt s}.  The elimination forms of sums and products
for the term context, {\tt split} and {\tt choose} are standard.

The {\tt return} primitive finalizes the branch and returns its body
as a suspension.  The branch is used by {\tt mfun\_rec} or {\tt mfun},
to key the memo table. If the result is found in the memo table, then
the suspension is disregarded and the result is re-used; otherwise the
suspension is forced and the result is stored in the memo table keyed
by the branch.  The {\tt mfun\_rec} primitive takes a recursive
function {\tt f} as a parameter and ``memoizes'' {\tt f} by
associating it with a memo table.  A subtle issue is that {\tt f} must
calls its memoized version recursively.  Therefore {\tt f} must take
its memoized version as a parameter.  Note also that the memoized
function internally converts its parameter to a resource before
applying {\tt f} to it.

The implementation described here does not check for correct usage.
To incorporate the run-time checks for correct usage, we need a more
sophisticated definition of resources in order to detect when a
resource is exposed out of its context (\ie, function instance).  In
addition, the interface must be updated so that the first parameter of
{\tt letBang}, {\tt letx}, and {\tt mcase}, occurs in suspended form.
This enables updating the state consisting of certain flags before
forcing a term.

\begin{figure}[h]
\centering{\fbox{\parbox{6in}{
\begin{codeListingS}
XX\=XX\=XX\=XX\=XX\=XX\=XX\=XX\=XX\=XX\=XX\=XX\=XXXXXXXXXXXXXXXXXXX\=\kill
structure Examples = \\
struct \\ 
\>type 'a box = 'a Box.box\\[2mm]

\>fun iB v = bang (fn i => i) v\\
\>fun bB b = bang (fn b => Box.keyOf b) b\\[2mm]

\>(** Boxed lists **)\\
\>datatype 'a blist' = NIL~|~CONS of ('a * (('a blist') box))\\
\>type 'a blist = ('a blist') box\\[2mm]

\>(** Hash-cons **)\\
\>fun hCons' (x') = letx (expose x') (fn (h',t') =>\\
~~~~~~~~~~~~~~~~~~~~letBang (expose h') (fn h => letBang (expose t') (fn t => \\
~~~~~~~~~~~~~~~~~~~~return (fn()=> box (CONS(h,t))))))\\
\>fun hCons x = mapply (mfun hCons') x \\[2mm]

\>(** Fibonacci **)\\
\>fun mfib' f (n') = \\
\>\>letBang (expose n') (fn n => \\
\>\>return (fn()=>if n < 2 then n else (mapply f (iB(n-1))) + (mapply f (iB(n-2))) \\
\>fun  mfib n = mapply (mfun\_rec mfib') n\\[2mm]

\>(** Knapsack **)\\
\>fun mks' mks (arg) = \\
\>\>letx (expose arg) (fn (c',l') =>\\
\>\>letBang (expose c') (fn c => \\
\>\>letBang (expose l') (fn l => return (fn () =>\\
\>\>\>case (unbox l) of \\
\>\>\>\>NIL => 0\\
\>\>\>|~CONS((w,v),t) => if (c < w) then mapply mks (pair (iB c) (bB t))\\
\>\>\>\>\>~~~~~~~~~~~~~~~else let val v1 = mapply mks (pair (iB c) (bB t))\\
\>\>\>\>\>\>~~~~~~~~~~~~~~~~~~~~~~val v2 = v + mapply mks (pair (iB (c-w)) (bB t))\\
\>\>\>\>\>~~~~~~~~~~~~~~~~~~~~in if (v1 > v2) then v1 else v2 end))))\\
\>val mks x = mfun\_rec mks'\\[2mm]

\>(** Quicksort **)\\
\>fun mqs () =\\
\>\>let val empty = box NIL\\
\>\>\>\>val hCons = mfun hCons'\\
\>\>\>\>fun  fil f l =\\
\>\>\>\>\>case (unbox l) of\\
\>\>\>\>\>\>NIL => empty\\
\>\>\>\>~~|~CONS(h,t) => if (f h) then (mapply hCons (pair (iB h) (bB (fil f t))))\\
\>\>\>\>\>\>~~~~~~~~~~~~~else (fil f t)\\
\>\>\>\>fun qs' qs (l') = letBang (expose l') (fn l => return (fn () => \\
\>\>\>\>\>\>case (unbox l) of\\
\>\>\>\>\>\>\>NIL => nil\\
\>\>\>\>\>~~|~CONS(h,t) => let val ll = fil (fn x=>x<h) t\\
\>\>\>\>\>\>\>\>~~~~~~~~~~~~~~~val gg = fil (fn x=>x>=h) t\\
\>\>\>\>\>\>\>\>~~~~~~~~~~~~~~~val sll = mapply qs (bB ll)\\
\>\>\>\>\>\>\>\>~~~~~~~~~~~~~~~val sgg = mapply qs (bB gg)\\
\>\>\>\>\>\>\>\>~~~~~~~~~~~in sll@(h::sgg) end))\\
\>\>in mfun\_rec qs' end\\
end
\end{codeListingS}}}}
\caption{Examples from \secref{framework} in the SML library.}
\label{fig:examples}
\end{figure}

\figref{examples} shows the examples from \secref{framework} written
in the SML library.  The examples assume a \ttt{Box} structure that
ascribes to the \ttt{BOX} signature (\figref{sig}).  The \ttt{hCons}
function follows the description closely.  The Fibonacci function
\ttt{mfib} applies its argument to a fresh memoized instance of the
Fibonacci function (\ttt{mfib'}).  As a results \ttt{mfib} allocates a
memo table every time it is called.  Since this table is shared by all
calls to \ttt{mfib'}, \ttt{mfib} runs in linear time.  When {\tt mfib}
finishes, this table can be garbage collected.  As with \ttt{mfib},
the Knapsack function \ttt{mks} also applies its argument to a fresh
memoized instance of the memoized Knapsack function (\ttt{mks'}).
Therefore, when \ttt{mks} returns, the allocated memo table can be
garbage collected.  For Quicksort, we provide a function {\tt mqs}
that returns an instance of memoized Quicksort when applied.  Each
such instance has its own memo table.  Note also that {\tt mqs}
creates a local instance of the hash-cons function so that each
instance of memoized Quicksort has its own memo table for
hash-consing---this table can be garbage collected when \ttt{mqs}
returns.

In the examples, we do not use the sum types provided by the library
to represent boxed lists, because ML sum types suffice for the
considered examples.  In general, one will use the provided sum types
instead of their ML counterparts (for example if an {\tt mcase} is
requires).  The examples in~\figref{examples} can be implemented using
the following definition
of boxed lists.\\
\parbox{3in}{
\begin{codeListingNormal}
XX\=XX\=XX\=XX\=XX\=XX\=XX\=XX\=XX\=XX\=XX\=XX\=XXXXXXXXXXXXXXXXXXX\=\kill
datatype 'a boxlist' = ROLL of (unit, (('a, 'a boxlist' box) prod)) sum\\
type 'a boxlist = ('a boxlist') box
\end{codeListingNormal}
}\\
Changing the code in~\figref{examples} to work with this definition
of boxed lists requires several straightforward modifications.

\section{Discussion}
\label{sec:discussion}

\paragraph{Space and cache management.}
Our framework associates a separate memo table with each memoized
function.  This allows the programmer to control the life-span of memo
tables by conventional scoping.  This somewhat coarse degree of
control is sufficient in certain applications such as in dynamic
programming, but finer level of control may be desirable for
applications where result re-use is less regular.  Such an application
can benefit from specifying a caching scheme for individual memo
tables so as to determine the size of the memo table and the
replacement policy.  We discuss how the framework can be extended to
associate a cache scheme with each memo table and maintain the memo
table accordingly.

The caching scheme should be specified in the form of a parameter to
the {\tt mfun} construct.  When evaluated, this construct will bind
the caching scheme to the memo table and the memo table will be
maintained accordingly.  Changes to the operational semantics to
accommodate this extension is small.  The store $\sigma$ will now map
a label to a pair consisting of a memo table and its caching scheme.
The handling of the {\tt return} will be changed so that the stores do
not merely expand but are updated according to the caching scheme
before adding a new entry.  The following shows the updated return
rule.  Here ${\cal S}$ denotes a caching scheme and $\theta$ denotes a
memo table.  The $\mathcd{update}$ function denotes a function that
updates the memo table to accommodate a new entry by possibly purging
an existing entry.  The programmer must ensure that the caching scheme
does not violate the integrity of the memo table by tampering with
stored values.\\

\smallskip
\parbox{\textwidth}{\centering
\renewcommand{\update}{{\mathcd{update}}}
\newcommand{\scheme}{{\cal S}}
\begin{tabular}{c}
{\mbox{$
\begin{array}{c}
\infer  [\quad\mbox{(Found)}]
        {\eis{\ms}{l}{\br}{\ret{\t}} \ereduces \eos{\v}{\ms}}
        {\ms(l) = (\theta,\scheme) & \theta(\br)=\v} \\[2ex]

\infer  [~\mbox{(Not Found)}]
        {\eis{\ms}{l}{\br}{\ret{\t}} \ereduces \eos{\v}{\msp[l \leftarrow
        \theta'']}}
        {       
         \begin{array}{c}
         \ms(l) = (\theta,\scheme) \quad \theta(\br)\uparrow \\
         \tis{\ms}{\t}  \treduces \tos{\v}{\msp}\\
         \msp(l) = (\theta',\scheme) \quad \theta'' = \update(\theta',\scheme,(\br,\v))
         \end{array}
         }

\end{array}
$}}
\end{tabular}}\\

\smallskip For example, we can specify that the memo table for the
Fibonacci function, shown in \figref{precise-dependences}, can contain
at most two entries and be managed using the least-recently-used
replacement policy.  This is sufficient to ensure that the memoized
Fibonacci runs in linear time.  This extension can also be
incorporated into the type system described in \secref{language}.
This would require that we associate types with memo stores and also
require that we develop a type system for ``safe'' $\mathcd{update}$
functions if we wish to enforce that the caching schemes are safe.

\smallskip
\paragraph{Local  versus non-local dependences.}
Our dependence tracking mechanism only captures ``local'' dependences
between the input and the result of a function.  A local dependence of
a function {\tt f} is one that is created inside the static scope of
{\tt f}.  A non-local dependence of {\tt f} is created when {\tt f}
passes its input to some other function {\tt g}, which examines {\tt
  f}'s input indirectly.  In previous work,
Abadi~\etal~\citep{AbadiLaLe96} and Heydon~\etal~\citep{HeydonLeYu00}
showed program analysis techniques for tracking non-local dependences
by propagating dependences of a function to its caller. They do not
discuss, however, efficiency implications of tracking non-local
dependences.  

Our framework can be extended to track non-local dependences by
introducing an application form for memoized functions in the
expression context.  This extension would, for example, allow for
dependences of non-constant length.  We chose not to support non-local
dependences because it is not clear if its utility exceeds its
efficiency effects.

\section{Conclusion}

We present language techniques for applying memoization selectively
under programmer control.  The approach makes explicit the performance
effects of memoization and yields programs whose running times can be
analyzed using standard techniques. A key aspect of the framework is
that it can capture both control and data dependences between input
and the result of a memoized function.  We show that the approach
accepts a relatively simple implementation by giving an implementation
as a library for the Standard ML language.  The main contributions of
the paper are the particular set of primitives we suggest and the
semantics along with the proofs that it is sound.  We expect that the
techniques can be implemented in any purely-functional language.


\bibliographystyle{abbrvnat}
\bibliography{main}

\begin{thebibliography}{40}
\providecommand{\natexlab}[1]{#1}
\providecommand{\url}[1]{\texttt{#1}}
\expandafter\ifx\csname urlstyle\endcsname\relax
  \providecommand{\doi}[1]{doi: #1}\else
  \providecommand{\doi}{doi: \begingroup \urlstyle{rm}\Url}\fi

\bibitem[Abadi et~al.(1996)Abadi, Lampson, and L\'{e}vy]{AbadiLaLe96}
M.~Abadi, B.~W. Lampson, and J.-J. L\'{e}vy.
\newblock Analysis and caching of dependencies.
\newblock In \emph{International Conference on Functional Programming}, pages
  83--91, 1996.

\bibitem[Acar(2005)]{Acar05}
U.~A. Acar.
\newblock \emph{Self-Adjusting Computation}.
\newblock PhD thesis, Department of Computer Science, Carnegie Mellon
  University, May 2005.

\bibitem[Acar and Ley-Wild(2009)]{AcarLW09}
U.~A. Acar and R.~Ley-Wild.
\newblock Self-adjusting computation with {Delta ML}.
\newblock \emph{Lecture Notes in Computer Science}, 5832/2009:\penalty0 1--38,
  2009.

\bibitem[Acar et~al.(2002)Acar, Blelloch, and Harper]{AcarBlHa02}
U.~A. Acar, G.~E. Blelloch, and R.~Harper.
\newblock Adaptive functional programming.
\newblock In \emph{Proceedings of the 29th Annual {ACM} Symposium on Principles
  of Programming Languages}, pages 247--259, 2002.

\bibitem[Acar et~al.(2003)Acar, Blelloch, and Harper]{AcarBlHa03}
U.~A. Acar, G.~E. Blelloch, and R.~Harper.
\newblock Selective memoization.
\newblock In \emph{Proceedings of the 30th Annual {ACM} Symposium on Principles
  of Programming Languages}, 2003.

\bibitem[Acar et~al.(2006{\natexlab{a}})Acar, Blelloch, Blume, Harper, and
  Tangwongsan]{AcarBlBlHaTa06}
U.~A. Acar, G.~E. Blelloch, M.~Blume, R.~Harper, and K.~Tangwongsan.
\newblock A library for self-adjusting computation.
\newblock \emph{Electronic Notes in Theoretical Computer Science}, 148\penalty0
  (2), 2006{\natexlab{a}}.

\bibitem[Acar et~al.(2006{\natexlab{b}})Acar, Blelloch, Blume, and
  Tangwongsan]{AcarBlBlTa06}
U.~A. Acar, G.~E. Blelloch, M.~Blume, and K.~Tangwongsan.
\newblock An experimental analysis of self-adjusting computation.
\newblock In \emph{Proceedings of the ACM Conference on Programming Language
  Design and Implementation}, 2006{\natexlab{b}}.

\bibitem[Acar et~al.(2007)Acar, Blume, and Donham]{AcarBlDo07}
U.~A. Acar, M.~Blume, and J.~Donham.
\newblock A consistent semantics of self-adjusting computation.
\newblock In \emph{European Symposium on Programming}, 2007.

\bibitem[Acar et~al.(2009)Acar, Blelloch, Blume, Harper, and
  Tangwongsan]{AcarBlBlHaTa09}
U.~A. Acar, G.~E. Blelloch, M.~Blume, R.~Harper, and K.~Tangwongsan.
\newblock An experimental analysis of self-adjusting computation.
\newblock \emph{ACM Trans. Prog. Lang. Sys.}, 32\penalty0 (1):\penalty0
  3:1--3:53, 2009.

\bibitem[Aho et~al.(1974)Aho, Hopcroft, and Ullman]{AhoHoUl74}
A.~V. Aho, J.~E. Hopcroft, and J.~D. Ullman.
\newblock \emph{The Design and Analysis of Computer Algorithms}.
\newblock Addison-Wesley, 1974.

\bibitem[Allen(1978)]{Allen78}
J.~Allen.
\newblock \emph{Anatomy of LISP}.
\newblock McGraw Hill, 1978.

\bibitem[Appel and Gon\c{c}alves(1993)]{AppelGo93}
A.~W. Appel and M.~J.~R. Gon\c{c}alves.
\newblock Hash-consing garbage collection.
\newblock Technical Report CS-TR-412-93, Princeton University, Computer Science
  Department, 1993.

\bibitem[Bellman(1957)]{Bellman57}
R.~Bellman.
\newblock \emph{Dynamic Programming}.
\newblock Princeton Univ. Press, 1957.

\bibitem[Bird(1980)]{Bird80}
R.~S. Bird.
\newblock Tabulation techniques for recursive programs.
\newblock \emph{ACM Computing Surveys}, 12\penalty0 (4):\penalty0 403--417,
  Dec. 1980.

\bibitem[Cohen(1983)]{Cohen83}
N.~H. Cohen.
\newblock Eliminating redundant recursive calls.
\newblock \emph{ACM Transactions on Programming Languages and Systems},
  5\penalty0 (3):\penalty0 265--299, July 1983.

\bibitem[Cook and Launchbury(1997)]{CookLa97}
B.~Cook and J.~Launchbury.
\newblock Disposable memo functions.
\newblock In \emph{Proceedings of Haskell Workshop}, 1997.

\bibitem[Cormen et~al.(1990)Cormen, Leiserson, and Rivest]{CormenLeRi90}
T.~H. Cormen, C.~E. Leiserson, and R.~L. Rivest.
\newblock \emph{Introduction to Algorithms}.
\newblock MIT Press/McGraw-Hill, 1990.

\bibitem[Demers et~al.(1981)Demers, Reps, and Teitelbaum]{DemersReTe81}
A.~Demers, T.~Reps, and T.~Teitelbaum.
\newblock Incremental evaluation of attribute grammars with application to
  syntax-directed editors.
\newblock In \emph{Principles of Programming Languages}, pages 105--116, 1981.

\bibitem[Goto and Kanada(1976)]{GotoKa76}
E.~Goto and Y.~Kanada.
\newblock Hashing lemmas on time complexities with applications to formula
  manipulation.
\newblock In \emph{Proceedings of the 1976 ACM Symposium on Symbolic and
  Algebraic Computation}, pages 154--158, 1976.

\bibitem[Hammer et~al.(2009)Hammer, Acar, and Chen]{HammerAcCh09}
M.~A. Hammer, U.~A. Acar, and Y.~Chen.
\newblock {CEAL}: a {C}-based language for self-adjusting computation.
\newblock In \emph{Proceedings of the 2009 ACM SIGPLAN Conference on
  Programming Language Design and Implementation}, June 2009.

\bibitem[Heydon et~al.(2000)Heydon, Levin, and Yu]{HeydonLeYu00}
A.~Heydon, R.~Levin, and Y.~Yu.
\newblock Caching function calls using precise dependencies.
\newblock In \emph{Proceedings of the 2000 ACM SIGPLAN Conference on
  Programming Language Design and Implementation}, pages 311--320, 2000.

\bibitem[Hilden(1976)]{Hilden76}
J.~Hilden.
\newblock Elimination of recursive calls using a small table of randomly
  selected function values.
\newblock \emph{BIT}, 16\penalty0 (1):\penalty0 60--73, 1976.

\bibitem[Hughes(1985)]{Hughes85}
R.~J.~M. Hughes.
\newblock Lazy memo-functions.
\newblock In \emph{Proceedings 1985 Conference on Functional Programming
  Languages and Computer Architecture}, 1985.

\bibitem[Ley-Wild et~al.(2008)Ley-Wild, Fluet, and Acar]{Ley-WildFlAc08}
R.~Ley-Wild, M.~Fluet, and U.~A. Acar.
\newblock Compiling self-adjusting programs with continuations.
\newblock In \emph{Proceedings of the International Conference on Functional
  Programming}, 2008.

\bibitem[Liu and Stoller(1999)]{LiuSt99}
Y.~A. Liu and S.~D. Stoller.
\newblock Dynamic programming via static incrementalization.
\newblock In \emph{European Symposium on Programming}, pages 288--305, 1999.

\bibitem[Liu et~al.(1998)Liu, Stoller, and Teitelbaum]{LiuStTe98}
Y.~A. Liu, S.~Stoller, and T.~Teitelbaum.
\newblock Static caching for incremental computation.
\newblock \emph{{ACM} Transactions on Programming Languages and Systems},
  20\penalty0 (3):\penalty0 546--585, 1998.

\bibitem[McCarthy(1963)]{McCarthy63}
J.~McCarthy.
\newblock A basis for a mathematical theory of computation.
\newblock In P.~Braffort and D.~Hirschberg, editors, \emph{Computer Programming
  and Formal Systems}, pages 33--70. North-Holland, Amsterdam, 1963.

\bibitem[Michie(1968)]{Michie68}
D.~Michie.
\newblock ``{M}emo'' functions and machine learning.
\newblock \emph{Nature}, 218:\penalty0 19--22, 1968.

\bibitem[Mostov and Cohen(1985)]{MostowCo85}
J.~Mostov and D.~Cohen.
\newblock Automating program speedup by deciding what to cache.
\newblock In \emph{Proceedings of the Ninth International Joint Conference on
  Artificial Intelligence}, pages 165--172, Aug. 1985.

\bibitem[Murphy et~al.(2002)Murphy, Harper, and Crary]{MurphyHaCr02}
T.~Murphy, R.~Harper, and K.~Crary.
\newblock The wizard of {TILT}: Efficient(?), convenient and abstract type
  representations.
\newblock Technical Report CMU-CS-02-120, School of Computer Science, Carnegie
  Mellon University, Mar. 2002.

\bibitem[Norvig(1991)]{Norvig91}
P.~Norvig.
\newblock Techniques for automatic memoization with applications to
  context-free parsing.
\newblock \emph{Computational Linguistics}, pages 91--98, 1991.

\bibitem[Pennings(1994)]{Pennings94}
M.~Pennings.
\newblock \emph{Generating Incremental Attribute Evaluators}.
\newblock PhD thesis, University of Utrecht, Nov. 1994.

\bibitem[Pennings et~al.(1992)Pennings, Swierstra, and Vogt]{PenningsSwVo92}
M.~Pennings, S.~D. Swierstra, and H.~Vogt.
\newblock Using cached functions and constructors for incremental attribute
  evaluation.
\newblock In \emph{Seventh International Symposium on Programming Languages,
  Implementations, Logics and Programs}, pages 130--144, 1992.

\bibitem[{Peyton Jones}(1987)]{Peyton-Jones87}
S.~{Peyton Jones}.
\newblock \emph{The Implementation of Functional Programming Languages}.
\newblock Prentice-Hall, 1987.

\bibitem[Pfenning(1995)]{Pfenning95}
F.~Pfenning.
\newblock Structural cut elimination.
\newblock In D.~Kozen, editor, \emph{Proceedings of the Tenth Annual Symposium
  on Logic in Computer Science}, pages 156--166. Computer Society Press, 1995.

\bibitem[Pfenning and Davies(2001)]{pfenning+:judgmental}
F.~Pfenning and R.~Davies.
\newblock A judgmental reconstruction of modal logic.
\newblock \emph{Mathematical Structures in Computer Science}, 11:\penalty0
  511--540, 2001.

\bibitem[Polakow and Pfenning(1999)]{PolakowPf99}
J.~Polakow and F.~Pfenning.
\newblock Natural deduction for intuitionistic non-commutative linear logic.
\newblock In J.-Y. Girard, editor, \emph{Proceedings of the 4th International
  Conference on Typed Lambda Calculi and Applications (TLCA'99)}, pages
  130--144. Springer-Verlag LNCS 1581, 1999.

\bibitem[Pugh(1988)]{Pugh88}
W.~Pugh.
\newblock \emph{Incremental Computation via Function Caching}.
\newblock PhD thesis, Department of Computer Science, Cornell University,
  August 1988.

\bibitem[Pugh and Teitelbaum(1989)]{PughTe89}
W.~Pugh and T.~Teitelbaum.
\newblock Incremental computation via function caching.
\newblock In \emph{Principles of Programming Languages}, pages 315--328, 1989.

\bibitem[Spitzen and Levitt(1978)]{SpitzenLe78}
J.~M. Spitzen and K.~N. Levitt.
\newblock An example of hierarchical design and proof.
\newblock \emph{Communications of the ACM}, 21\penalty0 (12):\penalty0
  1064--1075, 1978.

\end{thebibliography}


\end{document}